\documentclass[centering,11pt,reqno]{amsart}

\usepackage[utf8]{inputenc}
\usepackage[T1]{fontenc}
\usepackage{amsmath,amsthm}
\usepackage{amsfonts,amssymb}
\usepackage[mathscr]{eucal}
\usepackage{url}
\usepackage{mdframed}
\usepackage{paralist}
\usepackage[colorlinks,cite color=blue,pagebackref=true,pdftex]{hyperref}
\usepackage[margin=2.5cm]{geometry}
\usepackage{tikz}
\usetikzlibrary{calc,shapes}
\usepackage{mathtools}
\usepackage{multicol}
\usepackage{ytableau}
\usepackage{blkarray}
\usepackage[capitalize]{cleveref}
 \usepackage[foot]{amsaddr}
\usepackage{tikz-cd}
\usepackage{ytableau}
\usepackage{bbm}
\usepackage{comment}
\usepackage{color}
\usepackage{todonotes}
\usepackage{accsupp}

\newtheorem{theorem}{Theorem}[section]
\newtheorem*{theorem*}{Theorem}
\newtheorem*{cor*}{Corollary}
\newtheorem*{conj*}{Conjecture}
\newtheorem*{lemma*}{Lemma}
\newtheorem*{prop*}{Proposition}
\newtheorem{lem}[theorem]{Lemma}
\newtheorem{cor}[theorem]{Corollary}

\theoremstyle{definition}

\newtheorem{remark}[theorem]{Remark}

\newenvironment{claimproof}[1][\proofname]{\proof[#1]}{\endproof}

\renewcommand{\emptyset}{\varnothing}
\renewcommand{\epsilon}{\varepsilon}

\newcommand{\NP}{\textsc{NP}}
\newcommand{\PP}{\textsc{P}}
\title{Finding Short Paths On Simple Polytopes}

\author[Black]{Alexander E. Black}
\address{Department of Mathematics, Bowdoin College}
\email{a.black@bowdoin.edu}

\author[Steiner]{Raphael Steiner}
\address{Department of Mathematics, ETH Z\"{u}rich}
\email{raphaelmario.steiner@math.ethz.ch}
\thanks{Research of R.S. supported by SNSF Ambizione Grant No. 216071.}
\begin{document}

\maketitle

\begin{abstract}  
We prove that computing a shortest monotone path to the optimum of a linear program over a simple polytope is \NP-hard, thus resolving a 2022 open question of De Loera, Kafer, and Sanit\`{a}. As a consequence, finding a shortest sequence of pivots to an optimal basis with the simplex method is \NP-hard. In fact, we show this is \NP-hard already for fractional knapsack polytopes. By applying an additional polyhedral construction, we show that computing the diameter of a simple polytope is \NP-hard, resolving a 2003 open problem by Kaibel and Pfetsch. Finally, on the positive side, we show that every polytope has a small, simple extended formulation for which a linear length path may be found between any pair of vertices in polynomial time building upon a result of Kaibel and Kukharenko.
\end{abstract}

\section{Introduction}

Understanding the worst-case performance of the simplex method for linear programming across all choices of pivot rules is a longstanding research program established first with Dantzig's 1947 invention, with foundational contributions made across theoretical computer science, operations research, and combinatorics communities. Breakthroughs on the positive side include the polynomial average case analysis of Borgwardt \cite{b87}, the polynomial smoothed analysis by Spielman and Teng \cite{ST04}, and polynomial time versions for special families such as Orlin's network simplex algorithm \cite{OrlinFlows}. In the worst-case, the best known bound in terms of the number of inequalities and number of variables is subexponential originally due to Kalai \cite{k92} with follow up work improving the bounds in \cite{hz15}.

On the negative side, essentially all well-studied pivot rules are known to have superpolynomial worst case performance \cite{km72, jer73, AC78, GS79, Mur80, g83, k92,msw96, jour/cm/AZ98, conf/stoc/FHZ11, conf/ipco/Friedmann11,hz15, disser2020exponential, dissermosis, black2024exponentiallowerboundspivot, disserSTACS}. Pivot rules can even encode hard problems during their execution \cite{NPmighty, NPHardDuringExecution1, NPHardDuringExecution2}. Furthermore, the longstanding Hirsch conjecture that the diameter of the vertex-edge graph of a polytope is at most the number of inequalities minus the number of variables was disproven by Santos in \cite{HirschSolution}. This is a small sample of breakthroughs related to the nearly 80 years of consistent work dedicated to understanding this problem, yet fundamental questions remain open. 

Given a polytope $P = \{\mathbf{x} \in \mathbb{R}^{d}: A \mathbf{x} \leq \mathbf{b}\}$, defined by a constraint matrix $A\in \mathbb{R}^{m \times d}$ and right-hand side $\mathbf{b}\in \mathbb{R}^m$, it has a set of feasible bases consisting of the set of linearly independent subsets $B$ of rows of $A$ of size $d$ such that $A_{B}^{-1}\mathbf{b}_{B} \in P$, where $A_{B}$ and $\mathbf{b}_{B}$ denotes the matrix and right hand side restricted to the rows indexed by $B$. Two feasible bases $B$ and $B'$ are called adjacent if $|B \Delta B'| = 2$, which yields a graph associated to the polytope that we call the \textbf{feasible basis graph}. The simplex method solves a linear program by walking from basis to basis along the feasible basis graph. For a linear program $\max_{\mathbf{x}\in P}\mathbf{c}^{\intercal} \mathbf{x}$, the step from a feasible basis $B$ to a new feasible basis $B' = (B \setminus \{i\}) \cup \{j\}$ for some $i\in B, j\notin B$ is called \textbf{monotone} if the ray defined by 
\[\{\mathbf{x} \in \mathbb{R}^{d}: A_{B\setminus \{i\}} \mathbf{x} = \mathbf{b}_{B\setminus \{i\}}, A_{i}\mathbf{x}\leq \mathbf{b}_{i}\}\] 
is increasing with respect to $\mathbf{c}$. 

A monotone move along a single edge in the feasible basis exchange graph is called a \textbf{pivot}, and the run-time of the simplex method depends on the number of pivots taken to reach an optimum as well as the time to compute each pivot. 

There are several different pivot rules for the simplex method that have been studied. One that is particularly fundamental is the ``omniscient pivot rule,'' which simply chooses a shortest sequence of pivots to the optimum. Despite so many years of study, it is open whether this pivot rule may be computed in polynomial time. That is, given a linear program and a feasible initial basis, can one find a shortest monotone path in the feasible basis graph to an optimal basis in polynomial time? As our first main result, we prove that the answer is no assuming $\PP \neq \NP$. Concretely we show that the following decision problem is \NP-hard:

\medskip

\begin{mdframed}[innerleftmargin=0.5em, innertopmargin=0.5em, innerrightmargin=0.5em, innerbottommargin=0.5em, userdefinedwidth=0.95\linewidth, align=center]
	{\textsc{Pivot-distance}}
	\sloppy

	\noindent
	\textbf{Input:} A linear program $\max_{\mathbf{x}\in P}\mathbf{c}^{\intercal} \mathbf{x}$ defined by an objective vector $\mathbf{c} \in \mathbb{Q}^{d}$ and a polytope $P = \{\mathbf{x}\in \mathbb{R}^d\colon A\mathbf{x}\leq \mathbf{b}\}$ defined by a matrix $A\in \mathbb{Q}^{m\times d}$ and a vector $\mathbf{b}\in \mathbb{Q}^m$, a feasible basis $B \subseteq [m]$ of $P$, and a number $k\in \mathbb{N}$. 

	\noindent
    \textbf{Decision:} Does there exist a monotone sequence of at most $k$ pivots from $B$ to a basis $B^{\ast}$ corresponding to an optimal solution of the linear program?
	% \textbf{Decision:} Does there exist a sequence $B=B_0,\ldots,B_\ell=B$ of feasible bases of $P$ with $\ell\le k$ and $|B_i\Delta B_{i-1}|=2$ for every $i\in [\ell]$?
\end{mdframed}

In fact, we show a stronger statement related to another line of research of which the aforementioned hardness result is an immediate consequence (Theorem~\ref{thm:monotone} below). Namely, a related graph to the feasible basis graph is the \textbf{graph} of the polytope defined by the vertices and edges of the polytope. Originally, in 1994, Frieze and Teng showed \cite{FriezeTeng1994} that computing the diameter of the graph of a (possibly highly degenerate) input polytope $P$, called the \textbf{combinatorial diameter} and denoted $\mathrm{diam}(P)$, is weakly \NP-hard. Then much later in 2018, Sanit\`{a} showed in \cite{Sanita18} that computing the combinatorial diameter of the fractional matching polytope is strongly \NP-hard. This result spurred a flurry of other results. For example, Wulf showed that computing the combinatorial diameter is $\Pi_{2}$-complete \cite{HarderthanNPHard}. Various hardness results are known in the setting \cite{MonotoneDiameterHard, PolymatroidsPathHard, CircuitPivotRules, ShortestPathNoConstantFactor}. For special polytopes from algebraic combinatorics, hardness results are known but where the input is no longer the system of inequalities defining the polytope \cite{GraphicalZonotopes, ReconfigAlternatingCycles}. Similar hardness results have also been shown in generalizations of polytope graphs~\cite{CircuitPivotRules, fixeddimcircs,SigncompCircuitHard}. 

However, until very recently, all known hardness results regarding shortest paths and diameters of polytopes with their inequality description as input were for \textbf{degenerate} polytopes for which the vertex-edge graph and feasible basis exchange graph do not coincide, since a single vertex may be represented by multiple feasible bases.  Polytopes for which these two graphs coincide are called \textbf{simple}, and they correspond to polytopes for which every vertex is defined by precisely dimension many tight inequalities. In \cite{CircuitPivotRules}, De Loera, Kafer, and Sanit\`{a} asked whether there exists a polynomial time algorithm to find shortest (monotone) paths in graphs of simple polytopes. Concurrently with and independently of the work presented in this paper, in a recent breakthrough Dorfer~\cite{associahedronhard} showed that computing distances between pairs of vertices on the associahedron is \NP-complete, which implies that computing shortest paths on simple polytopes is \NP-hard. As our second main result, we prove the same result through a reduction from a different, arguably significantly simpler, class of polytopes (certain fractional knapsack polytopes, obtained by intersecting a hypercube with a carefully chosen halfspace). Formally, we show that the following decision problem is \NP-hard:

\medskip

\begin{mdframed}[innerleftmargin=0.5em, innertopmargin=0.5em, innerrightmargin=0.5em, innerbottommargin=0.5em, userdefinedwidth=0.95\linewidth, align=center]
	{$k$-\textsc{Distance on simple polytopes}}
	\sloppy

	\noindent
	\textbf{Input:} A simple polytope $P = \{\mathbf{x}\in \mathbb{R}^d\colon Ax\leq \mathbf{b}\}$ defined by a matrix $A\in \mathbb{Q}^{m\times d}$ and a vector $\mathbf{b}\in \mathbb{Q}^m$, two vertices $\mathbf{x},\mathbf{y}$ of $P$ and some number $k\in\mathbb{N}$.

	\noindent
	\textbf{Decision:} Do $\mathbf{x}$ and $\mathbf{y}$ have distance at most $k$ in the graph of $P$?
\end{mdframed}

\medskip

\begin{theorem}
\label{thm:shortestpathsimple}
$k$-\textsc{Distance on simple polytopes} is \NP-hard.
\end{theorem}

While assuming $\PP\neq \NP$, both Dorfer's result~\cite{associahedronhard} and Theorem~\ref{thm:shortestpathsimple} independently answer the aforementioned question of de Loera, Kafer and Sanit\`{a} in the negative, there are two further implications of our result which are not implied by that of Dorfer~\cite{associahedronhard}. First, one can easily find a path of length at most $O(\sqrt{m})$ between any pair of vertices on the associahedron in strongly polynomial time (see Lemma 2 of \cite{sleatortarjanthurston}), where $m$ denotes the number of facets. Thus, Dorfer's result could only imply at most that $O(\sqrt{m})$-distance is \NP-hard. In contrast, our argument shows that checking whether there exists a path of length at most $d-1$ in a $d$-dimensional simple polytope with $2d+1$ facets is NP-hard, so we have the following corollary:

\begin{cor}\label{mdcor}
$(m-d-2)$-\textsc{Distance on simple polytopes} is \NP-hard.
\end{cor}

In particular, unless $\PP = \NP$, finding a path on a simple polytope shorter than the Hirsch bound $m-d$ by more than $2$ cannot be done in polynomial time. 

However, the second and most important distinction between our Theorem~\ref{thm:shortestpathsimple} and Dorfer's work is the fact that our proof extends to the \emph{monotone} setting. A path in the vertex-edge graph is called \textbf{monotone} if each step along the path increases the objective function. Under nondegeneracy, monotonicity corresponds exactly to pivoting in the simplex method, and hence this setting is particularly relevant in the optimization context and has been studied in several prior works. As our third main result, we show that the following problem is NP-hard:

\medskip

\begin{mdframed}[innerleftmargin=0.5em, innertopmargin=0.5em, innerrightmargin=0.5em, innerbottommargin=0.5em, userdefinedwidth=0.95\linewidth, align=center]
	{$k$-\textsc{Monotone-Distance on simple polytopes}}
	\sloppy

	\noindent
	\textbf{Input:} A linear program $\max_{\mathbf{x}\in P}\mathbf{c}^{\intercal}\mathbf{x}$ defined by an objective vector $\mathbf{c} \in \mathbb{Q}^{d}$ and a simple polytope $\mathbf{x} \in P = \{\mathbf{x}\in \mathbb{R}^d\colon Ax\leq \mathbf{b}\}$ defined by a matrix $A\in \mathbb{Q}^{m\times d}$ and a vector $\mathbf{b} \in \mathbb{Q}^m$, a vertex $\mathbf{x}$ of $P$ and some number $k\in\mathbb{N}$.

	\noindent
	\textbf{Decision:} Is there a monotone path of length at most $k$ from $\mathbf{x}$ to a $\mathbf{c}$-maximizer?
\end{mdframed}

\begin{theorem}
\label{thm:monotone}
$(m-d-2)$-\textsc{Monotone Distance on simple polytopes} is \NP-hard.
\end{theorem}
Hence, unlike the results of Dorfer in \cite{associahedronhard}, our result implies the following, which is our first main result mentioned above.

\begin{cor}
\label{cor:minimalpivots}
\textsc{Pivot-distance} is \NP-hard. 
\end{cor}

The proofs of Theorem~\ref{thm:shortestpathsimple}, Corollary~\ref{mdcor} and Theorem~\ref{thm:monotone} will be presented in Section~\ref{sec:distance}.

Our fourth main result concerns a related problem, which appears as Problem~10 in the 2003 survey on polyhedral computation by Kaibel and Pfetsch \cite{KaibelPfetschSurvey}, where they ask for the complexity status of computing the combinatorial \emph{diameter} of a simple polytope. This problem was also reiterated by Sanit\`{a}~\cite{Sanita18} and Wulf~\cite{HarderthanNPHard}. By combining our aforementioned distance hardness result for simple polytopes with several additional ideas (that make up most of the technical work of this paper), we show that this problem, too, is \NP-hard. Concretely, we address the following decision problem.

\medskip

\begin{mdframed}[innerleftmargin=0.5em, innertopmargin=0.5em, innerrightmargin=0.5em, innerbottommargin=0.5em, userdefinedwidth=0.95\linewidth, align=center]
	{\textsc{Diameter of simple polytopes}}
	\sloppy

	\noindent
	\textbf{Input:} A simple polytope $P = \{\mathbf{x}\in \mathbb{R}^d\colon A\mathbf{x}\leq \mathbf{b}\}$ defined by a matrix $A\in \mathbb{Q}^{m\times d}$ and a vector $\mathbf{b}\in \mathbb{Q}^m$, and a number $k\in\mathbb{N}$.

	\noindent
	\textbf{Decision:} Does $\mathrm{diam}(P)\le k$ hold?
\end{mdframed}

\begin{theorem}
\label{thm:combodiam}
\textsc{Diameter of simple polytopes} is \NP-hard.
\end{theorem}

Our approach to proving Theorem~\ref{thm:combodiam} is to reduce \textsc{$k$-Distance on Simple Polytopes} to \textsc{Diameter of Simple Polytopes}. A priori, these are very different problems. To show that such a reduction nevertheless exists, we introduce and carefully analyze an intricate polyhedral construction (dubbed ``cyclic siloing'') which can be applied to any simple input polytope $P$ with a pair of vertices $\mathbf{u},\mathbf{v}$ to efficiently compute a larger simple polytope $Q$ whose diameter can be expressed as the sum of the distance of $\mathbf{u}$ and $\mathbf{v}$ on $P$ and another efficiently computable number $K$ (cf.~Theorem~\ref{thm:diameterprecise}). Given access to an oracle for \textsc{Diameter of simple polytopes}, one can then efficiently compute the distance of $\mathbf{u}$ and $\mathbf{v}$ on $P$ and solve \textsc{$k$-Distance on Simple Polytopes}. Our construction takes inspiration from a similar such construction previously analyzed in the context of lower bounds for the shadow simplex method~\cite{black2024exponentiallowerboundspivot}. We believe that the constructions introduced in this paper are of independent interest and will find applications to other problems in computational polytope theory.

At a very high level, our constructions resemble those in the aforementioned work of Frieze and Teng in~\cite{FriezeTeng1994}. In that work, they first construct a simple polytope by taking a linear programming relaxation of a combinatorial optimization problem and show that computing the radius, i.e. the furthest distance away from a given vertex in the graph of that polytope, is \NP-hard. They then apply a polyhedral construction to reduce diameter computation to the radius. However, our approach needs to overcome two major technical hurdles that stop Frieze and Teng's approach from working in our settings. First, we need a different construction in order to show finding shortest paths is \NP-hard instead of the radius. Our approach makes use of structural insights coming from understanding the geometric combinatorics of intersecting a hypercube with a halfspace, which was partly inspired by a similar construction in~\cite{PolymatroidsPathHard}. Second, the polyhedral construction used by Frieze and Teng to go from their hardness result for the radius to a hardness result for the diameter breaks simplicity. In particular, they iteratively cut off a vertex with a hyperplane (a process called \textbf{truncation}) and then take the convex hull with a new vertex close to that hyperplane (a process called \textbf{stacking}). Doing so repeatedly replaces a vertex with a tower separating that vertex from all of its neighbors. It breaks simplicity, because each vertex in the tower other than the top has more than $d$ neighbors. We instead perform another procedure that preserves simplicity by only applying truncations iteratively. In part, our approach is a refinement of the use of truncations by Holt and Klee in their study of Hirsch-sharp polytopes in \cite{Hirschsharp}.

The key idea behind our construction is to add $d$ truncations chosen purposefully to replace a vertex with a new vertex at which exactly the $d$ new added inequalities from the truncation are tight. Furthermore, we choose these truncations to never cut off any other vertex of the original polytope. Then, by construction, the new vertex is always at least $d$ steps away from any other vertex of the polytope before truncation. If we iterate this construction $r$ times the resulting new vertex is $r(d-1)+1$ steps away from any of the original vertices. This construction thus mimics the effect of building a tower like Frieze and Teng while preserving simplicity. We call this tower a \emph{cyclic silo}. To reduce distance computation to diameter computation, we replace the pair of vertices $\mathbf{u}$ and $\mathbf{v}$ we want to find a shortest path between with cyclic silos. In the resulting polytope, the pair of vertices at the tops of those towers will have distance precisely $2r(d-1)$ higher than the distance of $\mathbf{u}$ and $\mathbf{v}$ in the original polytope. One then aims to show that for $r$ sufficiently large, these vertices also attain the diameter. Therefore, computing the diameter of the resulting polytope allows one to find the distance between $\mathbf{u}$ and $\mathbf{v}$ in the original polytope, yielding the desired reduction. While this basic idea is approachable, to implement it in the desired way we navigate several intricate technical challenges. Namely, the choice of truncations, the efficient implementation of the construction, and a \emph{precise} rather than approximate control of the diameter of the resulting polytope turn out to be quite challenging. For the details, we refer to Section~\ref{sec:truncate}, where we carefully describe and analyze our constructions and discuss their technical challenges and how we overcome them.

Finally, all of these results presented so far are negative and indicate obstacles towards finding polynomial time simplex methods  conditional on $\PP \neq \NP$. Our fifth and final contribution is positive. In a recent work, Kaibel and Kukharenko~\cite{RockExtensions} showed that one can reduce the well-known open problem of solving linear programming in strongly polynomial time (often referred to as \emph{Smale's 9th problem} from his famous problem list for the 21st century) to instances where the feasible region forms a simpe polytope with combinatorial diameter bounded linearly in the number of inequalities. To prove this result, they introduce an operation they call a \textbf{rock extension}, which creates from a simple $d$-dimensional polytope with $m$ facets a closely related simple $(d+1)$-dimensional polytope with $m+1$ facets and the remarkable aforementioned property that its diameter is at most $2(m-d)$. Furthermore, these rock extensions have a distinguished vertex $(o,1)$ known as part of their construction. Their argument implies that there is a path from $(o,1)$ to any other vertex of length at most $m-d$, certifying the aforementioned diameter bound. In their work, they did not study the complexity of finding such a path. Here we show the following:
\begin{theorem} \label{thm:rockextension}
Let $Q$ be a rock extension with $m$ facets in $d$ dimensions. Let $\mathbf{u}$ and $\mathbf{v}$ be vertices of $Q$. Then one can find a path of length at most $2(m-d)$ from $\mathbf{u}$ to $\mathbf{v}$ in weakly polynomial time. If $(o,1)$ is taken as part of the input, a path of length at most $2(m-d)$ may be found in strongly polynomial time, and a path from $(o,1)$ to either vertex of length at most $m-d$ may also be found in strongly polynomial time.
\end{theorem}

This theorem follows from a very  simple analysis of the beautiful construction of Kaibel and Kukharenko in \cite{RockExtensions}. In Kukharenko's thesis \cite{kirllthesis}, he showed that the solution of the linear program $\min_{\mathbf{x}\in P}\mathbf{c}^{\intercal} \mathbf{x}$ is determined by the solution to the linear program $\min_{\mathbf{x}\in Q}(\mathbf{c},c_{z})^{\intercal}\mathbf{x}$, where $c_{z}$ may be computed in strongly polynomial time from $\mathbf{c}$. In that case, the path of length $m-d$ computed from $(o,1)$ to the optimum of the linear program is monotonically decreasing with respect to $(\mathbf{c}, c_{z})$. Our argument here implies that such a path may be computed in strongly polynomial time assuming the optimum of the linear program is known. More generally, it may be computed in weakly polynomial time by finding the optimum of that linear program.

This gives a weak sense in which there is indeed a weakly polynomial time simplex method. Namely, as a Phase 1 procedure, one implements the strongly polynomial time reduction to compute the rock extension and initializes at a vertex $(o,1)$. Then a path from $(o,1)$ to the optimum of $(\mathbf{c}, c_{z})$ of length at most $m-d$ may be computed in weakly polynomial time. However, of course, this is somewhat circular, since to compute this path we need to know the optimum, for which one has to appeal to a linear programming solver (however, possibly one quite different from the simplex method). At the same time, this tells us that complexity theory is not the obstruction to a polynomial time version of the simplex method with this Phase 1 procedure. In fact, assuming there is a strongly polynomial time algorithm for linear programming using \emph{any method}, there is a strongly polynomial algorithm to find a monotone path of length at most $m-d$ on a rock extension from $(o,1)$ to the optimum of $(\mathbf{c}, c_{z})$. In this sense, as a consequence of what we show here, there is a strongly polynomial time algorithm for linear programming if and only if there is a strongly polynomial time simplex method in a wide sense. This is a similar status to that of so-called \emph{circuit augmentation schemes} for linear programming (a generalization of the simplex methods which allow moving along a more general set of directions) due to the very recent breakthrough result of Natura in his proof of the polynomial circuit diameter conjecture in \cite{StronglyPolyCircDiameter}. His result demonstrates that if one can solve linear programming in strongly polynomial time using any method, then one can find a sequence of almost quadratically many circuit augmentations to the optimum of a linear program in strongly polynomial time. 

\section{Shortest Paths}\label{sec:distance}
We prove Theorem~\ref{thm:shortestpathsimple}, Corollary~\ref{mdcor} and Theorem~\ref{thm:monotone} by reduction from the following problem.
\medskip

\begin{mdframed}[innerleftmargin=0.5em, innertopmargin=0.5em, innerrightmargin=0.5em, innerbottommargin=0.5em, userdefinedwidth=0.95\linewidth, align=center]
	{\textsc{Partition with even sum}}
	\sloppy

	\noindent
	\textbf{Input:} A vector $(b_{1}, b_{2},\dots, b_{d}) \in \mathbb{Z}^{d}_{>0}$ with $\beta:=\sum_{i=1}^d b_i/2 \in \mathbb{Z}$.
    
	\noindent
	\textbf{Decision:} Does there exist a subset $S \subseteq [d]$ such that
\[\beta = \sum_{i \in S} b_{i} = \sum_{j \in [d] \setminus S} b_{j} ~\,\,\,\, ? \]
\end{mdframed}
Note that \textsc{Partition with even sum} is equivalent to the usual Partition problem, as there is trivially no solution to Partition if $\beta \notin \mathbb{Z}$. Thus, it is \NP-hard (cf.~Problem 20 in \cite{KarpsList}).

Given an instance $\mathbf{b}=(b_1,\ldots,b_d)\in \mathbb{Z}_{>0}^d$ of \textsc{Partition with even sum}, we define an associated polytope $P_\mathbf{b}$ as follows, where we set $\beta := \sum_{i=1}^{d} b_{i}/2 \in \mathbb{Z}$:
\[P_{\mathbf{b}} := [0,1]^{d+2} \cap \left\{\mathbf{x} \in \mathbb{R}^{d+2}\bigg\vert \sum_{i=1}^{d} b_{i} x_{i} - \beta x_{d+1} + (\beta+1/2) x_{d+2} \leq \beta + 1/4\right\}  \]
In what follows, whenever the vector $\mathbf{b}$ is clear from context, we will denote by $\mathbf{w}$ the vector obtained from $\mathbf{b}$ by extending it with entries $-\beta$ and $\beta+1/2$. The is, we define $\mathbf{w} := (b_{1}, b_{2},\dots, b_{n}, -\beta, \beta + 1/2)$. Then, in particular, \[P_{\mathbf{b}} = [0,1]^{d+2} \cap \{\mathbf{x} \in \mathbb{R}^{d+2}\mid \mathbf{w}^{\intercal} \mathbf{x} \leq \beta + 1/4\}.\] 
In the following, we prove several basic properties about the polytope $P_\mathbf{b}$, one of which is that it is a simple polytope. These properties allow us to reduce \textsc{Partition  with even sum} to the problem of finding shortest paths between two vertices of $P_\mathbf{b}$. 
\begin{lem}
\label{lem:knapsacksimple}
For all $\mathbf{b} \in \mathbb{Z}^{d}_{>0}$ with $\sum_{i=1}^d b_i$ even, the polytope $P_{\mathbf{b}}$ is $(d+2)$-dimensional and simple. 
\end{lem}
\begin{proof}
One can observe directly from the definition that $P_\mathbf{b}$ contains $[0,1/3]^{d+2}$ as a subset and is thus full-dimensional, i.e., of dimension $d+2$.

Since $[0,1]^{d+2}$ is simple, any vertex of $P_{\mathbf{b}}$ contained in at least $d+3$ defining hyperplanes must be in the hyperplane:
\[H_{\mathbf{b}} = \{\mathbf{x} \in \mathbb{R}^{d+2}: \sum_{i=1}^{n} b_{i} x_{i} - \beta x_{d+1} + (\beta+1/2) x_{d+2} = \beta + 1/4\}.\]
and also be a vertex of $[0,1]^{d+2}$ and therefore be a $\{0,1\}$-vector in that hyperplane. Since $b_{i}, \beta \in \mathbb{Z}$ for all $i \in [n]$ and $\beta + 1/2 \in \mathbb{Z}[1/2]$, for any $S \subseteq [d+2]$ we have
\[\sum_{i \in S} w_{i} \in \mathbb{Z}[1/2], \]
where $\mathbb{Z}[1/2]$ denotes the set of rational numbers of the form $p/q$ where $q \in \{1,2\}$ and $p \in \mathbb{Z}$. Hence, since
$\beta + 1/4 \notin \mathbb{Z}[1/2]$, we have 
\[\sum_{i \in S} w_{i} \neq \beta + 1/4.\]
Therefore, $P_{\mathbf{b}}$ is simple. 
\end{proof}

Next, we give an explicit combinatorial description of the vertices of $P_\mathbf{b}$. This description works in general for intersecting a hypercube with a halfspace, so no special assumptions on the vector $\mathbf{w}$ are used in the proof of the next statement. In what follows, for a subset $S \subseteq [d+2]$, let $\mathbf{e}_{S} = \sum_{i \in S} \mathbf{e}_{i}$.

\begin{lem}
\label{lem:knapsackvertices}
The graph of $P_{\mathbf{b}}$ has vertex set $V_{1} \cup V_{2}$, where 
\begin{align*}
    V_{1} &= \left\{\mathbf{e}_S\bigg\vert S\subseteq [d+2] \text{ s. t.} \sum_{i \in S} w_{i} \leq \beta\right\} \text{ and} \\
    V_{2} &= \left\{\mathbf{e}_{S} + \frac{\beta+1/4-\sum_{i\in S} w_{i}}{w_{k}} \mathbf{e}_{k}\bigg\vert \sum_{i \in S} w_{i} < \beta + 1/4< \sum_{j \in S \cup \{k\}} w_{j} \text{ or }\sum_{i \in S} w_{i} > \beta + 1/4 > \sum_{j \in S \cup \{k\}} w_{j}\right\}.
\end{align*}
\end{lem}

\begin{proof}
Every vertex of $[0,1]^{d+2}$ that is in the halfspace $\mathbf{w}^{\intercal} \mathbf{x} \leq \beta + 1/4$ remains a vertex, since 
\[P_{\mathbf{b}} = [0,1]^{d+2} \cap  \{\mathbf{x} \in \mathbb{R}^{d+2}: \mathbf{w}^{\intercal} \mathbf{x} \leq \beta + 1/4\}\]
This encompasses every vertex in $V_{1}$. Every other vertex is given by the intersection of the hyperplane $\{\mathbf{x} \in \mathbb{R}^{d+2}: \mathbf{w}^{\intercal} \mathbf{x} = \beta + 1/4\}$ with an edge of $[0,1]^{d+2}$. All edges of the hypercube $[0,1]^{d+2}$ are spanned between $\mathbf{e}_{S}$ and $\mathbf{e}_{S} + \mathbf{e}_{k}$ for some $S \subseteq [d]$ and $k \in [d] \setminus S$. Then the claimed description of the remaining set of vertices $V_{2}$ is obtained by computing the intersection points of such edges with the hyperplane defined by $\mathbf{w}^{\intercal} \mathbf{x} = \beta + 1/4$. 
\end{proof}

In what remains, we will encode the vertices of $P_\mathbf{b}$ purely combinatorially by identifying vertices in $V_1$ with their corresponding sets $S$ and vertices in $V_2$ with the unique pair $(S,k)$ of a set $S\subseteq [d]$ and an element $k\in [d]\setminus S$ satisfying the inequalities in the definition of $V_2$. We describe the graph using this terminology.

\begin{figure}
    \centering
    \begin{tikzpicture}
        \draw[thick, black] (0,-2) -- (2,0) -- (0,2) -- (-2,0) -- cycle;
        \draw (0,-2) node[red, circle, fill, scale =.5] {};
        \draw (-2,0) node[red, circle, fill, scale =.5] {};
        \draw (0,2) node[red, circle, fill, scale =.5] {};
        \draw (2,0) node[red, circle, fill, scale =.5] {};
        \draw (0,-2.5) node {$S$};
        \draw (0,2.5) node {$S \cup \{i,j\}$};
        \draw (3, 0) node {$S \cup \{j\}$};
        \draw (-3,0) node {$S \cup \{i\}$};
        \draw[ultra thick, red, dashed] (-1.5,-1) -- (1.5, -1);
        \draw (0,-.5) node[red] {(c)};
        \draw[ultra thick, blue, dotted] (-1.5,1) -- (1.5, 1);
        \draw (0,1.5) node[blue] {(f)};
    \end{tikzpicture}
    \hspace{.5 cm}
    \begin{tikzpicture}
        \draw[thick, black] (0,-2) -- (2,0) -- (0,2) -- (-2,0) -- cycle;
        \draw (0,-2) node[red, circle, fill, scale =.5] {};
        \draw (-2,0) node[red, circle, fill, scale =.5] {};
        \draw (0,2) node[red, circle, fill, scale =.5] {};
        \draw (2,0) node[red, circle, fill, scale =.5] {};
        \draw (0,-2.5) node {$S$};
        \draw (0,2.5) node {$S \cup \{i,j\}$};
        \draw (3, 0) node {$S \cup \{j\}$};
        \draw (-3,0) node {$S \cup \{i\}$};
        \draw[ultra thick, magenta, dashed] (-1.5,-1.5) -- (1.5, 1.5);
        \draw (0,-.5) node[magenta] {(e)};
        \draw[ultra thick, green, dotted] (-1.5,-.75) -- (-1.5, .75);
        \draw (-1,0) node[green] {(d)};
    \end{tikzpicture}
    \caption{Depicted are the four different ways a hyperplane can slice two edges of a $2$-face of a hypercube, which gives rise to the notions (c), (d), (e), and (f) of adjacency in Lemma \ref{lem:edgetypes}. Note there are truly six ways this can occur, but the remaining two correspond to swapping $i$ and $j$ for edges of type (d) and (e).}
    \label{fig:edgetypes}
\end{figure}
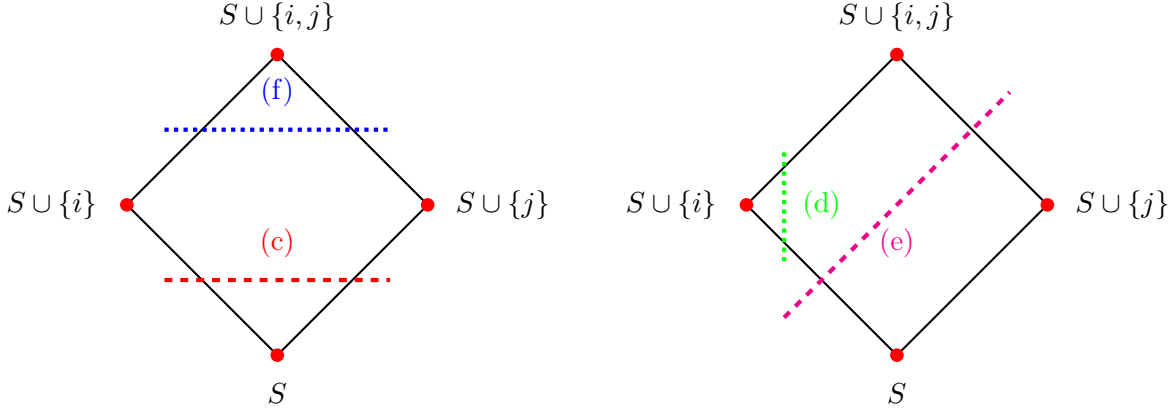

\begin{lem}
\label{lem:edgetypes}
Let $\mathbf{b}\in \mathbb{Z}_{>0}^d$ be such that $\sum_{i=1}^d b_i$ is even and let $\mathbf{u}$ and $\mathbf{v}$ be two vertices of $P_{\mathbf{b}}$. Then $\mathbf{u}$ and $\mathbf{v}$ are adjacent on $P_\mathbf{b}$ if and only if 
\begin{itemize}
    \item[(a)] $\mathbf{u} = S$ and $\mathbf{v} = T$ for some $S, T\subseteq [d+2]$ with $|S\Delta T|=1$, or
    \item[(b)]  $\mathbf{u} = S$ and $\mathbf{v} = (S,i)$ (for $i\notin S$) or $\mathbf{v} = (S\setminus \{j\}, j)$ (for $j\in S$) for some $S\subseteq [d+2]$, or
    \item[(c)] $\mathbf{u} = (S,i)$ and $\mathbf{v} = (S,j)$ for some $S\subseteq [d+2]$ and distinct $i,j\notin S$, or
    \item[(d)] $\mathbf{u} = (S,i)$ and $\mathbf{v} = (S \cup \{i\}, j)$ for some $S\subseteq [d+2]$ and distinct $i,j\notin S$, or
    \item[(e)] $\mathbf{u} = (S,i)$ and $\mathbf{v} = (S \cup \{j\}, i)$ for some $S\subseteq [d+2]$ and distinct $i,j\notin S$, or
    \item[(f)] $\mathbf{u} = (S \cup \{i\},j)$ and $\mathbf{v} = (S \cup \{j\}, i)$ for some $S\subseteq [d+2]$ and distinct $i,j\notin S$.
\end{itemize}
\end{lem}

\begin{proof}
Case (a) corresponds to adjacency on the hypercube, and two vertices in $V_{1}$ will be adjacent if and only if they are adjacent on the hypercube. 

Since $P_\mathbf{b}$ is simple, the only inequalities which are tight at a vertex $\mathbf{u}\in V_{1}$ are those coming from the hypercube $[0,1]^{d+2}$. It follows that $\mathbf{u}$ is adjacent to a vertex $\mathbf{v}\in V_{2}$ if and only if $\mathbf{v}$ is the intersection of an edge incident to the $V_{1}$ vertex on the hypercube $[0,1]^{d+2}$ with the hyperplane $\mathbf{w}^\intercal \mathbf{x}=\beta+1/4$. That is precisely what is captured by Case (b). 

All of cases (c), (d), (e), and (f) correspond to adjacency between vertices in $V_2$. Since the vertices in $V_2$ are those obtained as intersections of edges of the hypercube with the hyperplane $\mathbf{w}^\intercal\mathbf{x}=\beta+1/4$, the edges between them correspond exactly to the one-dimensional intersections of the two-dimensional faces $F$ of the hypercube $[0,1]^{d+2}$ with the hyperplane $\mathbf{w}^{\intercal} \mathbf{x} = \beta +(1/4)$. Furthermore, the vertices in $V_2$ connected by such an edge are the intersection points of two of the edges of the hypercube contained in $F$ with the hyperplane $\mathbf{w}^{\intercal} \mathbf{x} = \beta +(1/4)$. 

Note that the vertex-sets of the two-faces of the hypercube are exactly of the form $S$, $S \cup \{i\}$, $S \cup \{j\}$, $S \cup \{i,j\}$ where $S\subseteq [d+2]$ and $i,j\notin S$, and their four connecting edges are $[S, S\cup\{i\}]$, $[S, S\cup\{j\}]$, $[S\cup\{i\}, S\cup \{i,j\}]$ as well as $[S\cup\{j\}, S\cup \{i,j\}]$. Thus, for a fixed such two-dimensional face $F$ of the hypercube, there are up to $\binom{4}{2} = 6$ pairs of these four edges which could potentially be cut by the hyperplane and lead to adjacent vertices in $V_2$ on $P_\mathbf{b}$. Up to symmetry by swapping $i$ and $j$, there are really only four types of adjacency that can arise between vertices of $V_2$:

Case (c) corresponds to adjacency between the vertices of $P_{\mathbf{b}}$ coming from the edge from $S$ to $S \cup \{i\}$ and the edge from $S$ to $S \cup \{j\}$. Case $(d)$ comes from the edges $[S, S\cup \{i\}]$ and $[S \cup \{i\}, S \cup \{i,j\}]$. Case (e) comes from the edges $[S, S \cup \{i\}]$ and $[S \cup \{j\}, S \cup \{i,j\}]$. Finally case (f) comes from the pair $[S \cup \{i\}, S \cup \{i,j\}]$ and $[S \cup \{j\}, S \cup \{i,j\}]$. See Figure~\ref{fig:edgetypes} for a visualization of these cases.
\end{proof}

It turns out that the relevance of this characterization comes down to the following insight. If $S \subseteq T$ and $(S,i)$ and $(T,i)$ are both vertices, then the shortest a path between $(S,i)$ and $(T,i)$ in the graph of $P_\mathbf{b}$ could potentially be is $|T| - |S|$ by adding one element of $T$ to $S$ at a time. What we will prove is that if a shortest path of length $|T| - |S|$ exists, then it must be of that form, and that checking if such a path exists is \NP-hard by a reduction to \textsc{Partition with even sum}.

\begin{lem}
\label{lem:partitionreduction}
Let $\mathbf{b}=(b_1,\ldots,b_d)\in \mathbb{Z}_ {>0}^d$ such that $\sum_{i=1}^{d}{b_i}$ is even. Then
\begin{itemize}
    \item $(\emptyset,d+2)$ and $([d+1],d+2)$ are vertices of $P_{\mathbf{b}}$. 
    \item The shortest path between $(\emptyset,d+2)$ and $([d+1],d+2)$ is of length at most $d+1$ if and only if there exists a solution to \textsc{Partition  with even sum} with instance $\mathbf{b}$.
\end{itemize}

\end{lem}

\begin{proof}
Since 
\[\mathbf{w}^{\intercal} \mathbf{e}_{\emptyset} = 0 \leq \beta +1/4 < \beta +1/2 = \mathbf{w}^{\intercal} \mathbf{e}_{d+2},\] and 
\[\mathbf{w}^{\intercal} \mathbf{e}_{[d+1]} = \sum_{i=1}^{n} b_{i} -\beta = 2\beta - \beta = \beta < \beta + 1/4 < \beta + \beta +1/2 = \mathbf{w}^{\intercal} \mathbf{e}_{[d+2]},\]
$(\emptyset,d+2)$ and $([d+1],d+2)$ are vertices of $P_{\mathbf{b}}$ by the characterization of the vertices in Lemma \ref{lem:knapsackvertices}. This proves the first item of the lemma.

In the following, we will determine precisely the structure of the paths of length at most $d+1$ from $(\emptyset,d+2)$ to $([d+1],d+2)$ on $P_\mathbf{b}$, which will then yield the second item of the lemma.

From the characterization of edges in Lemma \ref{lem:edgetypes}, moving along any edge of $P_\mathbf{b}$ can only increase the size of the support of the current vertex by at most $1$. Thus, any path between $(\emptyset, d+2)$ and $([d+1], d+2)$ of length at most $d+1$ must in fact be of length \emph{exactly} $d+1$ and each step along the path must increase the size of the support by exactly $1$. The only edge types from Lemma \ref{lem:edgetypes} that increase the size of the support when we move along them starting from a vertex of the form $(S,i)$ are those of type (d) and (e). Since our path starts at $(\emptyset, d+2)$ and since moving along type (d) and (e) edges we stay within vertices of type $(S,i)$, it follows that any path of length $d+1$ from $(\emptyset,d+2)$ to $([d+1],d+2)$ on $P_\mathbf{b}$ must only use type (d) and (e) edges. We now claim that any such path in fact only uses type (e) edges. Indeed, towards a contradiction suppose it uses some type (d) edge and consider the earliest such edge along the path when starting from $(\emptyset, d+2)$. Since edges of type (e) always move from a vertex of the form $(S,i)$ to a vertex of the form $(S',i)$ and hence always preserve the ``second coordinate'', and since we start from the vertex $(\emptyset,d+2)$, the first edge of type (d) along the path must then start at a vertex of the form $(S,d+2)$ for some $S\subseteq [d+1]$ and go to $(S\cup \{d+2\},j)$ for some $j\notin S$ distinct from $d+2$. To have a total length of $d+1$, we would then need to reach $([d+1],d+2)$ from $(S\cup \{d+2\},j)$ using only type (d) and (e) edges which increase the support. However, this is impossible, since $S\cup \{d+2\}$ contains the element $d+2$ while $[d+1]$ does not, and since any type (d) and (e) edges used after will have to increase the support and hence preserve that $d+2$ is an element of the set in the tuple. Hence, we have reached the desired contradiction, and it follows that indeed all edges used along the path must be support-increasing edges of type (e).

Recalling the definition of type (e) edges, it now follows that every path of length at most $d+1$ from $(\emptyset, d+2)$ to $([d+1], d+2)$ on $P_\mathbf{b}$ must be of the form $(S_0,d+2),(S_1,d+2),\ldots,(S_{d+1},d+2)$ where \[\emptyset = S_{0} \subsetneq  S_{1} \subsetneq S_{2}, \dots \subsetneq S_{d+1} = [d+1]\] are such that $(S_{i},d+2)$ is a vertex of $P_\mathbf{b}$ and $S_{i} = S_{i-1} \cup \{k\}$ for some $k \in [d+1] \setminus S_{i-1}$ for all $1\le i \le d+1$.

We claim that such a sequence of sets exists (and hence the distance from $(\emptyset,d+2)$ to $([d+1],d+2)$ is at most $d+1$) if and only if there is a solution to \textsc{Partition with even sum}. Suppose first that such a sequence of sets exists. Let $i$ be minimal such that $d+1 \in S_{i}$. Then, since $(S_{i}, d+2)$ is a vertex of $P_{\mathbf{b}}$,

\begin{align*}
 -\beta +\sum_{j \in S_{i-1}} b_{j} &= \mathbf{w}^{\intercal}\mathbf{e}_{d+1} + \mathbf{w}^{\intercal} \mathbf{e}_{S_{i-1}}\\
 &= \mathbf{w}^{\intercal} \mathbf{e}_{S_{i}} \\
 &\leq \beta + 1/4 \\
 &\leq \mathbf{w}^{\intercal}\mathbf{e}_{S_{i} \cup \{d+2\}} \\
 &= \mathbf{w}^{\intercal} \mathbf{e}_{d+2} + \mathbf{w}^{\intercal} \mathbf{e}_{S_{i}} \\
 &=  \beta + 1/2 - \beta + \sum_{j \in S_{i-1}} b_{j} \\
 &= 1/2 + \sum_{j \in S_{i-1}} b_{j}.      
\end{align*}

In particular, $\beta + 1/4 \leq 1/2 + \sum_{j \in S_{i-1}} b_{j}$, so
\[\sum_{j \in S_{i-1}} b_{j} \geq \beta -1/4.\]

Similarly, since $(S_{i-1},d+2)$ is also a vertex of $P_{\mathbf{b}}$ and since $S_{i-1}$ does not contain $d+1$ by definition of $i$, we have:
\[\sum_{j \in S_{i-1}}b_{j} = \mathbf{w}^{\intercal}\mathbf{e}_{S_{i-1}} \leq \beta + 1/4.\]
It follows that 
\[\beta - 1/4 \leq \sum_{j \in S_{i-1}} b_{j} \leq \beta + 1/4\]
Since $b_{i} \in \mathbb{Z}$ for all $i \in [n]$ and $\beta\in \mathbb{Z}$, it follows that $\sum_{j \in S_{i-1}} b_{j} = \beta$. Therefore, in that case, \textsc{Partition with even sum} has a solution.

Suppose instead that \textsc{Partition with even sum} has a solution. Up to reordering we may without loss of generality assume then that
\[\sum_{i=1}^{k} b_{i} = \beta.\]
For each $j\in [d+1]$, define 
\[S_{j} =\begin{cases}
    \{1, \dots,j\} &\text{if } j \leq k \\ \\
    \{1,\dots, j-1\} \cup \{d+1\} &\text{if } j\ge k+1. 
\end{cases}\]
Then it suffices to show that $(S_{j}, d+2)$ is a vertex for each $j \in [d+1]$. If $j \leq k$, then 
\[\mathbf{w}^{\intercal} \mathbf{e}_{S_{j}} = \sum_{i \in S_{j}} b_{i} = \sum_{i=1}^{j} b_{i} \leq \sum_{i=1}^{k} b_{i} \leq \beta < \beta + 1/4 < \beta + 1/2 + \sum_{i\in S_{j}} b_{i} = \mathbf{w}^{\intercal} \mathbf{e}_{S_{j} \cup \{d+2\}}.\]
Hence, $(S_{j}, d+2)$ is a vertex in that case. 

If $j = k+1$, then
\[\mathbf{w}^{\intercal} \mathbf{e}_{S_{j}} = -\beta + \sum_{i=1}^{k} b_{i} = 0 < \beta + 1/4 < \beta + 1/2 + 0 = w_{d+2} + \sum_{i \in S_{j}} w_{i} = \mathbf{w}^{\intercal} \mathbf{e}_{S_{j} \cup \{d+2\}}.\]

Finally, suppose that $j \geq k+2$. Then $\sum_{i \in S_{j}} w_{i} \geq \sum_{i\in S_{k+1}}w_i=0$, and
\[\mathbf{w}^{\intercal}\mathbf{e}_{S_{j}} = \sum_{i \in S_{j}} w_{i} \leq \sum_{i \in [d+1]} w_{i} = \beta < \beta + 1/4 < \beta + 1/2 \leq \beta + 1/2 + \sum_{i \in S_{j}} w_{i} = \mathbf{w}^{\intercal} \mathbf{e}_{S_{j} \cup \{d+2\}}.\]
Hence, in all cases, $(S_{j}, d+2)$ is a vertex and so there is a path from $(\emptyset, d+2)$ to $([d+1], d+2)$ of length at most $d+1$ of the desired form. This concludes the proof of the equivalence claimed in the second item of the lemma. 
\end{proof}

This lemma yields Theorem \ref{thm:shortestpathsimple} as an immediate consequence.

\begin{proof}[Proof of Theorem \ref{thm:shortestpathsimple} and Corollary~\ref{mdcor}]
By Lemma \ref{lem:partitionreduction}, one can solve \textsc{Partition with even sum} by deciding whether the distance between two specified vertices of  $P_{\mathbf{b}}$ is at most $d+1$. Since $P_{\mathbf{b}}$ may be constructed from $\mathbf{b}$ in polynomial time, and its encoding length (in inequality description) is polynomially tied to the encoding length of the input $\mathbf{b}$ of \textsc{Partition with even sum}, and since $P_{\mathbf{b}}$ is simple by Lemma \ref{lem:knapsacksimple}, it follows that \textsc{$k$-Distance on simple polytopes} is \NP-hard when setting $k=d+1$. Since this equals $(2d+5)-(d+2)-2$ which is the number of defining inequalities of $P_\mathbf{b}$ minus the dimension of $P_\mathbf{b}$ minus two, this also proves Corollary~\ref{mdcor}.
\end{proof}

We next extend our result to the monotone setting and prove Theorem~\ref{thm:monotone}. Recall that $\mathbf{e}_{S} = \sum_{i \in S} \mathbf{e}_{i}$ for any $S \subseteq [n]$. 

\begin{lem}
\label{lem:monotone}
Let $\mathbf{b} = (b_{1}, \dots, b_{d}) \in \mathbb{Z}^{d}_{> 0}$ such that $\sum_{i=1}^{d} b_{i} = 2\beta$ is even. Let $\varepsilon:=\frac{1}{5\beta}$ and $\mathbf{c} = \mathbf{e}_{[d+1]} + \varepsilon \mathbf{e}_{d+2}$. Then 
\begin{itemize}
    \item $([d+1], d+2)$ is the unique $\mathbf{c}$-maximum.
    \item If $S \subsetneq T \subseteq [d+1] $, then $(S, d+2)$ has objective value less than $(T,d+2)$.
\end{itemize}
\end{lem}

\begin{proof}
By Lemma \ref{lem:partitionreduction}, $([d+1], d+2)$ is a vertex of $P_{\mathbf{b}}$. Furthermore, \[\mathbf{w}^{\intercal} \mathbf{e}_{[d+2]} = \sum_{i=1}^{d+2} \mathbf{w}_{i} = \sum_{i=1}^{n} b_{i} -\beta + (\beta +1/2) = 2\beta + 1/2 > \beta + 1/4.\]
Hence, $\mathbf{e}_{[d+2]} \notin P_{\mathbf{b}}$. Let $\mathbf{v}$ be the vector corresponding to $([d+1], d+2)$. Then by Lemma \ref{lem:knapsackvertices},
\[\mathbf{v} = \mathbf{e}_{[d+1]} + \alpha \mathbf{e}_{d+2}\]
 for some $\alpha > 0$. It follows that 
\[\mathbf{c}^{\intercal}\mathbf{v} = d+1 + \varepsilon\alpha. \]
Any other vertex is of the form $\mathbf{e}_{S} + \alpha' \mathbf{e}_{i}$, where $S\subsetneq [d+2]$, $i\notin S$ and $0 \leq \alpha' < 1$. In particular, by Lemma \ref{lem:knapsackvertices}, if $\alpha' > 0$, then
\[\alpha' = \frac{\beta + 1/4 - \sum_{j \in S} w_{j}}{w_{i}}.\]
Suppose first that $i \neq d+2$. Then $S\cap [d+1]$ is a proper subset of $[d+1]$, and $w_{i}$ is integral with $|w_{i}| \leq \beta$. Furthermore, $4w_{j} \in \mathbb{Z}$ for each $j \in S$. Note that $\alpha'=0$ or $0<\alpha'<1$. In the latter case, we have
\[\alpha' =\frac{\beta + 1/4 - \sum_{j \in S} w_{j}}{w_{i}}=\frac{\left|4\beta + 1 - \sum_{j \in S} 4w_{j}\right|}{4\left|w_{i}\right|} \leq 1 - \frac{1}{4|w_i|} < 1 - \frac{1}{5 \beta}.\]
Since $\varepsilon = \frac{1}{5\beta}$, it follows that $\alpha' + \varepsilon < 1$ (and this clearly also holds in the case $\alpha'=0$). Consequently,
\[\mathbf{c}^{\intercal}(\mathbf{e}_{S} + \alpha' \mathbf{e}_{i}) \le (|S \cap [d+1]|+\varepsilon) + \alpha' \le d + \alpha'+\varepsilon<d+1+\varepsilon\alpha= \mathbf{c}^{\intercal}\mathbf{v}.\] 
For the second case, suppose, $i = d+2$ (and hence $S\subseteq [d+1]$). We then obtain
\[\mathbf{c}^{\intercal}(\mathbf{e}_{S} + \alpha' \mathbf{e}_{d+2}) = |S| + \varepsilon\alpha' \leq  \max\{d + \varepsilon\alpha',d+1\}<d+1+\varepsilon\alpha = \mathbf{c}^{\intercal}\mathbf{v} \] where in the second step we used that $\mathbf{e}_S+\alpha'\mathbf{e}_i\neq \mathbf{v}$
meaning that $S \neq [d+1]$ or $\alpha'=0$. Thus, $\mathbf{v}$ is the unique $\mathbf{c}$-maximizer, as desired. This concludes the proof of the first item of the lemma.

For the second item, consider any $S \subsetneq T \subseteq [d+1]$. Then for any $0  < \alpha < 1$ and $0 < \beta < 1$, $\mathbf{c}^{\intercal}(\mathbf{e}_{S} + \alpha \mathbf{e}_{d+2})= |S| + \varepsilon\alpha < |S|+1 \leq |T| < \mathbf{c}^{\intercal}(\mathbf{e}_{T} + \beta \mathbf{e}_{d+2})$. It follows that $(S,d+2)$ has lower objective value than $(T,d+2)$, as desired.
\end{proof}

This lemma allows us to immediately extend our result to the monotone setting.

\begin{proof}[Proof of Theorem \ref{thm:monotone}]
Note that by Lemma \ref{lem:partitionreduction} one can solve \textsc{Partition with even sum} by checking whether a path from $(\emptyset, d+2)$ to $([d+1],d+2)$ of length at most $d+1$ exists. By Lemma \ref{lem:monotone}, $([d+1],d+2)$ is the optimum of the objective $\mathbf{c}$ from the statement of Lemma \ref{lem:monotone}. From the proof of Lemma \ref{lem:partitionreduction}, a path of length at most $d+1$ exists if and only if a path exists of the form \[(S_{0}, d+2), (S_{1},d+2),\dots,(S_{d+1}, d+2),\]
where $S_{i} \subsetneq S_{i+1}$ for each $i \in [0,n]$. By Lemma \ref{lem:monotone}, that path is increasing with respect to $\mathbf{c}$. Hence, a $\mathbf{c}$-increasing path of length at most $d+1$ from $(\emptyset,d+2)$ to $([d+1],d+2)$ exists if and only if there is a path of length at most $d+1$ from $(\emptyset,[d+2])$ to $([d+1],d+2)$. This is true if and only if there is a solution to \textsc{Partition with even sum}. Hence, the same reduction works, showing that monotone distance on simple polytopes is \NP-hard. 
\end{proof}

\section{Diameters}\label{sec:truncate}

Throughout this section, we only consider simple polytopes of dimension $d\ge 3$. Furthermore, we will always assume that we only work with irredundant inequality descriptions of our polytopes. In particular, we assume that every vertex of our polytopes satisfies exactly $d$ of the defining inequalities with equality. Furthermore, we will always assume that the entries of the matrix and the right-hand side defining our polytope have rational entries.  This is crucial for some of our statements and lemmas, even though it will not always be explicitly mentioned. We will also throughout use the notation $d_P(\mathbf{u},\mathbf{v})$ to denote the (combinatorial) distance between two vertices $\mathbf{u}, \mathbf{v}$ in the graph of a polytope $P$. 

Suppose we are given a $d$-dimensional simple polytope $P$ described by $m$ inequalities and a vertex $\mathbf{v}$ of $P$. We can then compute the $d$ neighbors $\mathbf{v}_1,\ldots,\mathbf{v}_d$ of $\mathbf{v}$ on $P$ and ``cut $\mathbf{v}$ off'' from each of these neighbors by adding a single new inequality. Namely, we may compute the mid-points $\mathbf{m}_i:=\frac{\mathbf{v}+\mathbf{v}_i}{2}, i=1,\ldots,d$ of the incident edges of $\mathbf{v}$ and then compute the unique hyperplane passing through the points $\mathbf{m}_1,\ldots,\mathbf{m}_d$. It is easy to see that this hyperplane separates $\mathbf{v}$ from all other vertices of the polytope. Finally, we add a new inequality to $P$ describing the halfspace of this hyperplane which does not contain $\mathbf{v}$. This operation is called \textbf{truncation} and yields a new simple polytope $T(P,\mathbf{v})$. The vertices of $T(P,\mathbf{v})$ are exactly those of $P$ except $\mathbf{v}$ plus the $d$ additional vertices $\mathbf{m}_1,\ldots,\mathbf{m}_d$. 

Since we will later need it for our reductions, let us record the following useful statement about computing and encoding repeated truncations of polytopes.
\begin{lem}\label{lem:complexity}
Suppose $P$ is a simple $d$-dimensional polytope with rational irredundant inequality description and with bit-encoding length $L$, and let $r\in \mathbb{N}$. Suppose we are given as input $P$ as well as a sequence of $r$ vertices which are  revealed to us during the process one at a time, and each time a new vertex is revealed to us we have to perform a truncation at this vertex. Then an inequality description of the final polytope $Q$ (obtained after performing the sequence of $r$ truncations) with encoding length $\mathrm{poly}(L,r)$ can be computed in time $\mathrm{poly}(L,r)$.
\end{lem}
\begin{proof}
Recall that we assume that the coefficients and constants of the inequalities defining $P$ are rational numbers, and hence the same is true for all vertices of $P$. Since taking midpoints keeps the coordinates of vectors rational, all new vertices constructed during the process are rational. In particular, the vertices of $Q$ are rational.

To start, we bound the encoding-lengths of any vertices appearing in the process polynomially in $L$ and in $r$. Consider first the vertices of $P$. Each such vertex is the solution of a linear equation system over a $d\times d$ invertible submatrix of the constraint matrix. Hence, by a standard application of Cramer's rule and Hadamard's inequality, the bit-encoding length of each vertex of $P$ is upper-bounded by $O(d^2L)$ each. Next observe that by definition of truncation, each new vertex obtained in one of the $r$ truncations to obtain $Q$ can be written as a convex combination of vertices of $P$ where all coefficients are in $\left\{0,\frac{1}{2^r},\ldots,\frac{2^r-1}{2^r},1\right\}$. Moreover, note that in this convex combination we only need to consider vertices of $P$ that at some point are either picked as the truncated vertex or a neighbor of it. Since in constructing $P$, we certainly consider at most $(d+1)\cdot r$ such vertices of $P$, it follows that each new vertex constructed at some point of the process is a convex combination of at most $(d+1)r$ vertices of $P$ with coefficients in $\{0,\frac{1}{2^r},\ldots,\frac{2^r-1}{2^r},1\}$. Recall that each vertex of $P$ has encoding length $O(d^2L)$ and in particular each entry of a vector in $P$ has numerator and denominator at most $2^{O(d^2L)}$. Hence, every entry of any vertex computed in the construction process for $Q$ can be written in the form
$$\sum_{i=1}^{(d+1)r}\frac{\alpha_i}{2^r}\cdot \frac{a_i}{b_i},$$ where $\alpha_i\in \{0,1\ldots,2^r\}$ and $|a_i|,|b_i|\le 2^{O(d^2L)}$ for every $i$. This equals $\frac{p}{q}$, where 
$$|q|=2^r\prod_{i=1}^{(d+1)r}{|b_i|}\le 2^{r+(d+1)r\cdot O(d^2L)}=2^{O(d^3Lr)},$$ and 
$$|p|= \left|\sum_{i=1}^{(d+1)r}\alpha_ia_i\prod_{j\neq i}b_j\right|\le \sum_{i=1}^{(d+1)r}2^r\cdot 2^{(d+1)r\cdot O(d^2L)}=2^{O(d^3Lr).}$$ Hence, every vertex computed in the construction process for $Q$ has encoding length at most $d\cdot O(d^3Lr)=O(d^4Lr)=O(L^3r)$, where we used that $L\ge d^2$ in the last step. 

It remains to argue that we can compute an inequality description of $Q$ with encoding length $\mathrm{poly}(L,r)$ in time $\mathrm{poly}(L,r)$. Let $\mathbf{s}_1,\ldots,\mathbf{s}_r$ be the vertices revealed to us one by one during the process, and suppose that for some $1\le j \le r$ we have already computed an inequality description of the polytope $P^{j-1}$ obtained from $P$ after performing truncations at vertices $\mathbf{s}_1,\ldots,\mathbf{s}_{j-1}$ in this order, and now a new vertex $\mathbf{s}_j$ of $P^{j-1}$ is revealed to us, at which we are supposed to perform the next truncation. To obtain the inequality description of the next polytope $P^{j}=T(P^{j-1},\mathbf{s}_j)$, we proceed as follows:

\begin{itemize}
    \item We first compute all the $d$ neighbors $\mathbf{u}_1^j,\ldots,\mathbf{u}_d^j$ of $\mathbf{s}_j$ on $P^{j-1}$, in polynomial time in the encoding length of $P^{j-1}$. Concretely, we can do this by trying out all possible base exchanges at $\mathbf{s}_j$ and thus solving up to $d(m(P^{j-1})-d)\le d(m(P)+r-d)\le d(L+r-d)$ linear equation systems (here $m(P^{j-1}), m(P)$ denote the number of inequalities describing $P^j$ and $P$, respectively) of size $d\times d$ whose coefficients form a submatrix of the constraint matrix of $P^{j-1}$. Hence, this can be executed in polynomial time in the encoding length of $P^{j-1}$.
    \item We then compute the $d$ midpoints $\frac{\mathbf{u}_1^j+\mathbf{s}_j}{2},\ldots,\frac{\mathbf{u}_d^j+\mathbf{s}_d}{2}$. This can be done in polynomial time in the encoding lengths of $\mathbf{u}_1^j,\ldots,\mathbf{u}_d^j$ and $\mathbf{s}_j$, and hence, by what we argued about the encoding lengths of these vectors above, in time $\mathrm{poly}(L,r)$.
    \item Finally, we compute the coefficients of the one new inequality to be added to $P^{j-1}$ to obtain $P^{j}$, i.e., a vector $\mathbf{w}$ and some $\alpha\in \mathbb{R}$ such that the hyperplane $\mathbf{w}^\intercal \mathbf{x}=\alpha$ passes through all the midpoints $\frac{\mathbf{u}_1^j+\mathbf{s}_j}{2},\ldots,\frac{\mathbf{u}_d^j+\mathbf{s}_d}{2}$. To do so, it suffices to find the (up to scaling unique) non-trivial solution to the $d\times (d+1)$-sized homogeneous linear equation system whose row vectors are obtained from $\frac{\mathbf{u}_1^j+\mathbf{s}_j}{2},\ldots,\frac{\mathbf{u}_d^j+\mathbf{s}_d}{2}$ by appending $1$-s at the end. Since by what we showed above also the encoding lengths of $\frac{\mathbf{u}_1^j+\mathbf{s}_j}{2},\ldots,\frac{\mathbf{u}_d^j+\mathbf{s}_d}{2}$ are polynomial in $L$ and in $r$, it follows again by using Cramer's rule and Hadamard's inequality, that we can compute a desired non-trivial solution $(\mathbf{w},\alpha)$ to this linear system whose encoding length is bounded by $\mathrm{poly}(L,r)$, and of course, it can be also computed in time $\mathrm{poly}(L,r)$ by solving the linear system.
\end{itemize}

By the last point, we find that each of the $r$ truncation steps increases the encoding length of the polytope by at most $\mathrm{poly}(L,r)$, and hence each of the $r$ polytopes $P^{1}, P^{2},\dots,P^{r} = Q$ built in the process has encoding length at most $\mathrm{poly}(L,r)\cdot r=\mathrm{poly}(L,r)$. With this knowledge, it follows that each of the three steps above can be executed in time $\mathrm{poly}(L,r)$ for each of the $r$ truncations, and hence computing the inequality description of the final polytope $Q$ also can be done in time $\mathrm{poly}(L,r)\cdot r=\mathrm{poly}(L,r)$. This establishes the desired statements and concludes the proof of the lemma.
\end{proof}

As an organizational tool to keep track of the impact of truncating repeatedly on the combinatorial structure of the polytope, we use a generating function. Namely, let a $d$-dimensional simple polytope $P$ with $m$ inequalities labeled by numbers $1, 2,\dots,m$ be given. Let $\mathcal{B}\subseteq \binom{[m]}{d}$ denote the set of feasible bases of $P$. Consider the polynomial ring $\mathbb{Z}[x_{1}, x_{2}, \dots, x_{m}]$, and for a subset $S \subseteq [m]$ let us denote $\mathbf{x}^{S} := \prod_{i \in S} x_{i}$. We now define the \textbf{generating function of feasible bases} of $P$ by 
\[f_{P}(\mathbf{x}) = \sum_{B \in \mathcal{B}} \mathbf{x}^{B}.\]
Now consider a vertex $\mathbf{v}^\ast$ of $P$. Now when we truncate $P$ at $\mathbf{v}^\ast$ to obtain $T(P,\mathbf{v}^\ast)$, we will associate a new variable $x_{m+1}$ in the polynomial ring with the added inequality. Our next lemma precisely describes how truncation changes the generating function.

\begin{lem}
\label{lem:firsttruncation}
Let $P$ be a $d$-dimensional simple polytope with $m$ facets labeled $1,2,\dots,m$. Let $\mathcal{B} \subseteq \binom{[m]}{d}$ denote the set of feasible bases of $P$. Let $\mathbf{v}^\ast$ be a vertex of $P$ and $B^{\ast}$ the corresponding feasible basis. Then we have
\[f_{T(P,\mathbf{v}^\ast)}(\mathbf{x}) = \sum_{B \in \mathcal{B} \setminus \{B^{\ast}\}} \mathbf{x}^{B} + \sum_{i \in B^{\ast}} \mathbf{x}^{B^\ast \setminus \{i\}} x_{m+1} = f_{P}(\mathbf{x}) - \mathbf{x}^{B^{\ast}} + \sum_{i \in B^{\ast}} \mathbf{x}^{B^{\ast} \setminus \{i\}} x_{m+1}. \]
\end{lem}

\begin{proof}
By definition, $T(P,\mathbf{v}^\ast)$ has one additional new inequality and so has $m+1$ inequalities. The truncation only removes precisely one vertex, namely $\mathbf{v}^\ast$. Thus, each feasible basis of $P$ other than $B^\ast$ remains a feasible basis of $Q$. Furthermore, a new feasible basis is also added for each new vertex. There are exactly $d$ new vertices in $T(P,\mathbf{v}^\ast)$ compared to $P$, namely those corresponding to the intersection of the $d$ edges of $P$ incident with $\mathbf{v}$ with the hyperplane defining the new inequality we added. Therefore, the feasible bases corresponding to new vertices are precisely of the form $(B^{\ast} \setminus \{i\}) \cup \{m+1\}$, where $i$ ranges through the elements of $B^\ast$. Putting these facts together yields the desired formula for the generating function of $T(P,\mathbf{v}^\ast)$. 
\end{proof}

Next, we would like to understand the effect of \emph{repeated} truncation on the generating function. As a first step, it will thus be convenient for us to reformulate the expression in Lemma~\ref{lem:firsttruncation} in the case that the polytope $P$ to which we apply the truncation already comes with two classes of inequalities, namely $m$ ``old'' inequalities associated with variables $x_1,\ldots,x_m$ in the generating function, and $k-1$ ``new'' inequalities associated with $k-1$ new variables $y_{1}, y_{2}, \dots, y_{k-1}$, which we think of arising from $k-1$ previous truncations. Here, we will use the convention that variable $y_{i}$ corresponds to the inequality added in the $i$th previous truncation. In particular, in the following $y_1,\ldots,y_k$ will take the role of the variables $x_{m+1},\ldots,x_{m+k}$ in the previous formulation of Lemma~\ref{lem:firsttruncation}. Our next corollary is simply a restatement of Lemma \ref{lem:firsttruncation} in this new set-up. To simplify notation we define, for each $S\subseteq [m]$ and $T\subseteq [k]$, the following shorthand: $$\mathfrak{C}_{k}(\mathbf{x}^{S}\mathbf{y}^{T}) :=\sum_{t \in T} \mathbf{x}^{S} \mathbf{y}^{(T \setminus \{t\} )\cup \{k\}}+\sum_{s \in S} \mathbf{x}^{S \setminus \{s\}} \mathbf{y}^{T \cup \{k\}}.$$

\begin{cor}\label{cor:truncateys}
Let $P$ be a $d$-dimensional simple polytope with facets in two classes of size $m$ and $k-1$, which are labeled $1,2,\dots,m$ and $1,2,\dots,(k-1)$, respectively. Let $\mathcal{B} \subseteq 2^{[m]} \times 2^{[k-1]}$ denote the set of feasible bases of $P$. Let $B^{\ast} = (S,T)$ be a feasible basis, with corresponding vertex $\mathbf{v}^\ast$. Then we have
\[f_{T(P,\mathbf{v}^\ast)}(\mathbf{x},\mathbf{y}) =  %\sum_{B \in \mathcal{B}\setminus\{B^\ast\}} (\mathbf{x},\mathbf{y})^{B} + \mathfrak{C}_k(\mathbf{x}^{S}\mathbf{y}^{T})= 
f_P(\mathbf{x},y_1,\ldots,y_{k-1})-\mathbf{x}^S\mathbf{y}^T + \mathfrak{C}_k(\mathbf{x}^{S}\mathbf{y}^{T}).\]
\end{cor}

Next, we will use this observation to describe the generating function of a new polytope constructed by a specific sequence of $d$ iterated truncations.

\begin{lem}\label{lem:function}
Let $P$ be a $d$-dimensional simple polytope with facets labeled $1, 2,\dots, m$ such that $[d]$ is a feasible basis. Then, given $P$ (in inequality description) and this feasible basis as input, we can, in polynomial time in the encoding length of $P$, construct a new simple polytope $S$ with $m + d$ facets, with corresponding variables $x_1,\ldots,x_m,y_1,\ldots,y_d$ and with feasible basis generating function 
\[f_{S}(\mathbf{x}, \mathbf{y}) = f_{P}(\mathbf{x}) - \mathbf{x}^{[d]} + \mathbf{y}^{[d]} + \sum_{k=0}^{d-1} \left(\sum_{i=1}^{k} \mathbf{x}^{[d] \setminus [k]} \mathbf{y}^{[k+1] \setminus \{i\}} + \sum_{j=k+2}^{d} \mathbf{x}^{[d] \setminus ([k] \cup \{j\})} \mathbf{y}^{[k+1]} \right).\] 
\end{lem}

\begin{proof}
By definition, we have for all $0 \leq k < d$: 
\begin{align*}
\mathfrak{C}_{k+1}(\mathbf{x}^{[d] \setminus [k]} \mathbf{y}^{[k]}) &=\sum_{i=1}^{k} \mathbf{x}^{[d] \setminus [k]}\mathbf{y}^{[k+1] \setminus \{i\}} + \sum_{j=k+1}^{d}\mathbf{x}^{[d] \setminus ([k] \cup \{j\})} \mathbf{y}^{[k+1]} \\
&= \mathbf{x}^{[d]\setminus [k+1]}\mathbf{y}^{[k+1]} + \sum_{i=1}^{k} \mathbf{x}^{[d] \setminus [k]}\mathbf{y}^{[k+1]\setminus \{i\}} + \sum_{j=k+2}^{d }\mathbf{x}^{[d] \setminus ([k] \cup \{j\})} \mathbf{y}^{[k+1]}.
\end{align*}
To construct the polytope $S$, we will construct a sequence $T^0,\ldots,T^{d}$ of polytopes, where we initialize $T^0:=P$, and for $k=0,\ldots,d-1$ we define $T^{k+1}:=T(T^{k},\mathbf{v}^{k})$, where $\mathbf{v}^{k}$ is defined as the vertex of $T^{k}$ corresponding to the the monomial $\mathbf{x}^{[d]\setminus[k]}\mathbf{y}^{[k]}$. Note that this is always a well-defined operation, since by the above calculation and by Corollary~\ref{cor:truncateys}, we can see that if the generating function of $T^{k}$ contains the monomial  $\mathbf{x}^{[d]\setminus[k]}\mathbf{y}^{[k]}$ then the generating function of the resulting polytope $T^{k+1}$ after truncation will contain the monomial $\mathbf{x}^{[d] \setminus [k+1]} \mathbf{y}^{[k+1]}$. Finally, we set $S:=T^{d}$. Then, by Corollary~\ref{cor:truncateys} and our above computation, the overall sum of the monomials added to the generating function across the whole procedure amounts to
\[ \sum_{k=0}^{d-1} \left(\mathbf{x}^{[d]\setminus [k+1]}\mathbf{y}^{[k+1]} + \sum_{i=1}^{k} \mathbf{x}^{[d] \setminus [k]}\mathbf{y}^{[k+1]\setminus \{i\}} + \sum_{j=k+2}^{d }\mathbf{x}^{[d] \setminus ([k] \cup \{j\})} \mathbf{y}^{[k+1]} \right).\]
Similarly, the overall sum of the monomials subtracted from the generating function across the whole procedure (cf.~Corollary~\ref{cor:truncateys}) equals $\sum_{k=0}^{d-1} \mathbf{x}^{[d] \setminus [k]}\mathbf{y}^{[k]}$. Subtracting that off yields 
\[f_{S}(\mathbf{x}, \mathbf{y}) = f_{P}(\mathbf{x}) - \mathbf{x}^{[d]} + \mathbf{y}^{[d]} + \sum_{k=0}^{d-1} \left(\sum_{i=1}^{k} \mathbf{x}^{[d] \setminus [k]} \mathbf{y}^{[k+1] \setminus \{i\}} + \sum_{j=k+2}^{d} \mathbf{x}^{[d] \setminus ([k] \cup \{j\})} \mathbf{y}^{[k+1]} \right),\] as desired. Finally, note that since $S$ arises from $P$ by a sequence of $d$ truncations, by Lemma~\ref{lem:firsttruncation} we can compute an inequality description of the final polytope $S$ with encoding length polynomial in the encoding length of $P$, in polynomial time in the encoding length of $P$.
\end{proof}

\begin{figure}
    \centering
    \[\begin{tikzpicture}
        \draw[black, thick] (-1,-1) -- (0,.5) -- (1,-1) -- (3,-2) -- (0,2) -- (-3,-2) -- (-1,-1) -- (1,-1) -- (-3,-2) -- (3,-2);
        \draw[black,thick] (0,.5) -- (0,2);
        \draw[black,thick] (-1,-1) -- (0,2);
        \draw[black, thick] (1,-1) -- (0,2);        
    \end{tikzpicture}
    \hspace{.5cm}
    \begin{tikzpicture}
        \draw[black, thick] (-1,-1) -- (0,.5) -- (1,-1) -- (3,-2) -- (0,2) -- (-3,-2) -- (-1,-1) -- (1,-1) -- (-3,-2) -- (3,-2);
        \draw[black,thick] (0,.5) -- (0,2);
        \draw[black,thick] (-1,-1) -- (0,2);
        \draw[black, thick] (1,-1) -- (0,2);
        \filldraw[red, opacity = .25] (-1,-1) -- (0,2) -- (0,.5) -- cycle;
        \filldraw[blue, opacity = .25] (-3,-2) -- (3,-2) -- (1,-1) -- cycle;
    \end{tikzpicture}
    \]
    \caption{Depicted is the silo construction in $d=3$ dimensions in the normal fan of the polytope. Namely, the outer triangle corresponds to the normal cone of the vertex being cut off. We visualize this as a triangle by slicing the cone with a plane. Then the siloing subdivides that slice. The basis exchange graph corresponds to the dual graph of the triangulation. In this picture it is already visible that two cells may be of distance $d=3$ away from each other as is the case for the highlighted cells on the right side of the picture. }
    \label{fig:triangulation}
\end{figure}
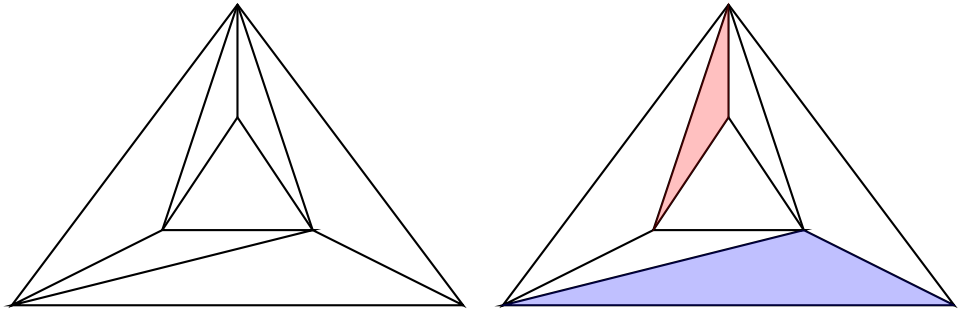

In the remainder of the paper, the construction of a simple polytope $S$ with $m+d$ facets starting from an ordered feasible basis of a simple polytope $P$ with $m$ facets as described in Lemma~\ref{lem:function} will be referred to as \textbf{siloing} due to its interpretation as building a tower to create an isolated vertex as depicted visually in the normal fan in Figure \ref{fig:triangulation}. Note that the assumption that the feasible basis corresponding to the vertex where we start truncating equals $[d]$ is not of any essence, and thus, more generally, can be applied to any \emph{ordered} feasible basis of a polytope by relabeling. Thus, given a simple polytope $P$, a vertex $\mathbf{v}$ of $P$ and a linear ordering $b_1\prec b_2\prec\cdots\prec b_d$ on the elements of the feasible basis $\{b_1,\ldots,b_d\}$ associated with $\mathbf{v}$, we use the notation $S(P,\mathbf{v},\prec)$ for the polytope $S$ obtained from Lemma~\ref{lem:firsttruncation} applied after relabeling such that $b_i$ is the $i$th inequality of the polytope $P$ for $i=1,\ldots,d$, and call it the \textbf{silo} of $P$ at $(\mathbf{v},\prec)$. We then always have that $S(P,\mathbf{v},\prec)$ can be computed in polynomial time given $P, \mathbf{v}$ and $\prec$ and satisfies a formula for the generating function corresponding to that of Lemma~\ref{lem:firsttruncation} after suitable relabeling. As further terminology, when constructing a silo $S(P,\mathbf{v},\prec)$, we say that the vertex $\mathbf{v}$ is \textbf{being siloed}. We also call the final vertex added in the construction process for $S(P,\mathbf{v},\prec)$ (concretely, the vertex whose feasible basis corresponds to the monomial $\mathbf{y}^{[d]}$) the \textbf{peak} of the silo. Overall, the construction effectively replaces the vertex $\mathbf{v}$ being siloed with a tower that peaks at the peak of the silo, much like in the reduction of Frieze and Teng's paper \cite{FriezeTeng1994}. 

Siloing is almost enough to achieve our goal of reducing the shortest path problem on simple polytopes to the problem of computing the diameter of simple polytopes: By design of the construction (and as will be formally verified later), given some vertex $\mathbf{u}$ of the original polytope $P$ distinct from the siloed vertex $\mathbf{v}$, the distance from $\mathbf{u}$ to the peak of $S(P,\mathbf{v},\prec)$ is precisely the distance from $\mathbf{u}$ to $\mathbf{v}$ in $P$ plus $d-1$. 

This property naturally leads to the following idea for our reduction: Namely, to apply the silo construction repeatedly. Concretely, suppose we are given as input a simple polytope $P$ and a pair of vertices $\mathbf{u}, \mathbf{v}$ between which we want to solve \textsc{$k$-Distance on Simple Polytopes}. Then we silo $\mathbf{u}$, silo at the peak of that silo, and keep siloing at peaks repeatedly $r$ times (for some suitably chosen, large enough, parameter $r$), such that the last peak $\mathbf{u}'$ we created will have distance at least $r(d-1)$ from any vertex of the original polytope. Then we repeat the same process at $\mathbf{v}$, yielding another ``last'' peak vertex $\mathbf{v}'$ with the same property. One can then check that in the graph of the resulting polytope $Q$, we will have found a new pair $\mathbf{u}', \mathbf{v}'$ of vertices satisfying $d_Q(\mathbf{u}',\mathbf{v}') = d_P(\mathbf{u},\mathbf{v}) + 2r(d-1)$. Furthermore, one can check that for any pair of vertices in the original polytope, even after this repeated siloing their distance will have changed by an additive constant of at most $6$. Hence, (provided $r$ was chosen large enough) the distance between any two original vertices of $P$ in the final polytope $Q$ will be much smaller than that of $\mathbf{u'}$ and $\mathbf{v}'$. The hope would thus be to show that $d_P(\mathbf{u}', \mathbf{v}')$ equals the diameter of the new polytope $Q$, such that we could reduce the problem of computing/bounding the distance between $\mathbf{u}, \mathbf{v}$ in the graph of $P$ to the problem of computing the diameter of $Q$, providing the desired hardness result claimed by Theorem~\ref{thm:combodiam}. 

However, this idea narrowly fails to work, at least in the simple form that we now described. On the one hand, one can check that for any two vertices in the same so-called tower of silos, their distance is at most $r(d-1) + 1$ and hence these vertices will be closer to each other than $\mathbf{u}'$ and $\mathbf{v}'$, as desired. However, the problem is that it may happen that the distance between two vertices in different towers may be $2r(d-1) + d_P(\mathbf{u}, \mathbf{v}) + 2$ in the worst case and thus (slightly) bigger than $d_Q(\mathbf{u}',\mathbf{v}')$. Hence, in such a case all we may conclude is that the diameter of $Q$ lies somewhere between $d(\mathbf{u},\mathbf{v}) + 2r(d-1)$ and $d(\mathbf{u},\mathbf{v}) + 2r(d-1) + 2$. This, unfortunately, is not quite enough to determine the shortest path distance between $\mathbf{u}$ and $\mathbf{v}$ exactly. One would need an APX-hardness result here, which we however have no access to. In fact, we leave finding such an APX-hardness result as an open problem in Section \ref{sec:conclusion}. 

Due to this subtle technical difficulty, we must be careful with how exactly we perform the described sequence of repeated siloings. To do so, we need to delve down into the combinatorics of the polytopes resulting from siloing and identify precisely which types of vertices can lead to the aforementioned increased distances. Based on this deeper understanding of the construction, we can construct the towers of silos such that the aforementioned bad situation never arises. 

\begin{figure}
    \centering

\begin{comment}
    \[\begin{tikzpicture}[scale=1,
  dot/.style={circle,fill=black,inner sep=1.2pt},
  diag/.style={very thick,gray},
  up/.style={thin},
  skip/.style={thick, dotted, blue}
]
\def\d{5} % <-- set d here

% optional: light grid
\draw[step=1,gray!25] (1,1) grid (\d,\d);

% optional: diagonal for reference
\draw[diag] (1,1) -- (\d,\d);

% axis ticks/labels (optional)
\foreach \k in {1,...,\d} {
  \node[below] at (\k,1) {\scriptsize \k};
  \node[left]  at (1,\k) {\scriptsize \k};
}

% points (deleted product)
\foreach \i in {1,...,\d} {
  \foreach \j in {1,...,\d} {
    \ifnum\i=\j\relax\else
      \node[dot] at (\i,\j) {};
    \fi
  }
}

% arrows upward: (i,j) -> (i,j+1), only if both endpoints are off-diagonal
\foreach \i in {1,...,\d} {
  \foreach \j in {1,...,\d} {
    \ifnum\j<\d\relax
      \pgfmathtruncatemacro{\jp}{\j+1}
      \ifnum\i=\j\relax\else
        \ifnum\i=\jp\relax\else
          \draw[up] (\i,\j) -- (\i,\jp);
        \fi
      \fi
    \fi
  }
}

% extra arrows skipping over the diagonal:
% (i,i-1) -> (i,i+1), defined only when i-1 >= 1 and i+1 <= d
\foreach \i in {1,...,\d} {
  \ifnum\i>1\relax
    \ifnum\i<\d\relax
      \pgfmathtruncatemacro{\im}{\i-1}
      \pgfmathtruncatemacro{\ip}{\i+1}
      % both endpoints are automatically off-diagonal; draw the arrow
      \draw[skip] (\i,\im) -- (\i,\ip);
    \fi
  \fi
}

\end{tikzpicture}\]
\end{comment}
 \includegraphics[scale=1]{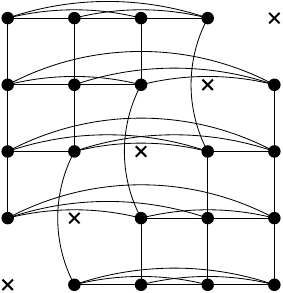}
    \label{fig:grid}
        \caption{Depicted is the graph $G_{d}$ for $d = 5$. Vertices of the same height (i.e., second coordinate) are pairwise adjacent. Otherwise, there is an edge from a vertex to the first vertex above it and below it that is in the graph.}
\end{figure}

To start making these high-level ideas more concrete, we first need to better understand the adjacencies between the new vertices after siloing. To do so, for a natural number $d$, we define a graph $G_d$ as having vertex set
$$V(G_d):=\{(a,b)\in [d]^2|a\neq b\}$$ and where two distinct vertices $(a,b)$ and $(a',b')$ with $b\le b'$ are adjacent if and only if one of the following holds:
\begin{itemize}
    \item $b=b'$, or
    \item $b'=b+1$, $a=a'$ and $b\neq a-1$, or
    \item $b'=b+2$, $a=a'$, and $b=a-1$.
\end{itemize}
We call $G_d$ the \textbf{$d$-th silo graph}. See Figure~\ref{fig:grid} for a visual illustration for $d=5$. The following lemma describes the adjacencies between new vertices of the silo of a $d$-dimensional polytope precisely in terms of the graph $G_d$.
\begin{lem}\label{lem:isomorphism}
Let $P$ be a $d$-dimensional simple polytope described by $m$ irredundant inequalities. Let $\mathbf{v}$ be a vertex of $P$ and let $\prec$ be a linear order on the elements $b_1\prec b_2\prec \cdots\prec b_d$ of the feasible basis defining $\mathbf{v}$. Let $H$ be the subgraph of the graph of $S(P,\mathbf{v},\prec)$ induced by the vertices in $S(P,\mathbf{v},\prec)$ that are not vertices of $P$ and distinct from the peak. Then $H$ and $G_d$ are isomorphic. Moreover, assuming that we label the inequalities as $1,\ldots,m$ such that $b_i$ receives label $i$ for $i=1,\ldots,d$, an isomorphism from $G_d$ to $H$ is given by mapping a vertex $(a,b)$ of $G_d$ to the vertex of $S(P,\mathbf{v},\prec)$ associated with the monomial $\mathbf{x}^{[d]\setminus ([b-1] \cup \{a\})}\mathbf{y}^{[b]}$ if $a>b$ and  $\mathbf{x}^{[d]\setminus [b-1]}\mathbf{y}^{[b]\setminus \{a\}}$ if $a<b$.
\end{lem}
\begin{proof}
Notice that by Lemma~\ref{lem:function} the vertices of $S(P,\mathbf{v},\prec)$ that are considered in $H$ are exactly those whose monomials in the generating function appear in the sum
\[\sum_{k=0}^{d-1} \left(\sum_{i=1}^{k} \mathbf{x}^{[d] \setminus [k]} \mathbf{y}^{[k+1] \setminus \{i\}} + \sum_{j=k+2}^{d} \mathbf{x}^{[d] \setminus ([k] \cup \{j\})} \mathbf{y}^{[k+1]} \right).\]
These are exactly the monomials of the form
$\mathbf{x}^{[d]\setminus ([b-1] \cup \{a\})}\mathbf{y}^{[b]}$ for some $1\le b<a\le d$ and $\mathbf{x}^{[d]\setminus [b-1]}\mathbf{y}^{[b]\setminus \{a\}}$ for some $1\le a<b\le d$, giving the desired bijection between vertices of $H$ and $G_d$. 

Now consider any distinct vertices $(a,b), (a',b')$ of $G_d$ with $b\le b'$. Let $\mathbf{v}, \mathbf{v'}$ be their associated distinct vertices in $H$. To prove the statement of the lemma, we have to show that $\mathbf{v}$ and $\mathbf{v}'$ are adjacent in $H$ if and only if $b=b'$; or $b'=b+1$, $a=a'$ and $b\neq a-1$; or $b'=b+2, a=a'$ and $b=a-1$. We start by showing sufficiency and split this into cases.

\textbf{Case~1.} Suppose first that $b=b'$. Then the monomials associated with $\mathbf{v}$ and $\mathbf{v}'$ are both obtained from $\mathbf{x}^{[d]\setminus [b-1]}\mathbf{y}^{[b]}$ by omitting exactly one variable. Hence, their bases have a symmetric difference of at most two and so $\mathbf{v}$ and $\mathbf{v}'$ are adjacent in the graph of $S(P,\mathbf{v},\prec)$ and hence also in $H$, as desired.

\textbf{Case~2.} Suppose next that $b'=b+1$, $a=a'$ and $b\neq a-1$. Then then we can obtain the monomial of $\mathbf{v}'$ from that of $\mathbf{v}$ by replacing the variable $x_{b}$ with the variable $y_{b+1}$ (note that since $a=a'\notin \{b,b+1\}$, the variable $x_{b}$ indeed always occurs in the monomial representing $\mathbf{v}$, and the variable $y_{b+1}$ indeed occurs in the monomial representing $\mathbf{v}'$). Hence, again the corresponding bases have a symmetric difference of size at most two and so $\mathbf{v}$ and $\mathbf{v}'$ are adjacent in $H$. 

\textbf{Case~3.} Finally suppose that $b'=b+2$, $a=a'$ and $b=a-1$. Then $\mathbf{v}$ and $\mathbf{v}'$ are represented by $\mathbf{x}^{[d]\setminus ([b-1]\cup \{b+1\})}\mathbf{y}^{[b]}$ and $\mathbf{x}^{[d]\setminus [b+1]}\mathbf{y}^{[b]\cup \{b+2\}}$, respectively. Since the latter can be obtained from the first by exchanging the variable $x_{b}$ for the variable $y_{b+2}$, indeed $\mathbf{v}$ and $\mathbf{v}'$ are adjacent also in this last case.

It remains to show necessity of the conditions. So suppose that $\mathbf{v}$ and $\mathbf{v}'$ are adjacent in $H$, i.e. their corresponding feasible bases of $S(P,\mathbf{v},\prec)$ have symmetric difference of size two, and let us prove that at least one of the three conditions for adjacency in $G_d$ is satisfied. Since the first condition holds if $b=b'$, in what follows we may and will assume $b'>b$. 

Let us denote by $M, M'\subseteq \{x_1,\ldots,x_d,y_1,\ldots,y_d\}$ the sets of variables occurring in the monomials representing $\mathbf{v}$ and $\mathbf{v}'$, respectively. We then have $$M=\{y_1,\ldots,y_b,x_b,\ldots,x_d\}\setminus \{s\}, M'=\{y_1,\ldots,y_{b'},x_{b'},\ldots,x_d\}\setminus \{s'\},$$ where $s\in \{x_a,y_a\}\cap \{y_1,\ldots,y_b,x_b,\ldots,x_d\}$ and $s'\in \{x_{a'},y_{a'}\}\cap \{y_1,\ldots,y_{b'},x_{b'},\ldots,x_d\}$. Since $\mathbf{v}$~and $\mathbf{v}'$ are adjacent, the sets $M$ and $M'$ must have symmetric difference exactly two. We therefore find
$$2=|M\Delta M'|=|(\{y_1,\ldots,y_b,x_b,\ldots,x_d\}\Delta \{s\})\Delta (\{y_1,\ldots,y_{b'},x_{b'},\ldots,x_d\}\Delta \{s'\})|$$
$$=|\{x_b,\ldots,x_{b'-1},y_{b+1},\ldots,y_{b'}\}\Delta (\{s\}\Delta \{s'\})|$$ $$=2(b'-b)+|\{s\}\Delta\{s'\}|-2|\{x_b,\ldots,x_{b'-1},y_{b+1},\ldots,y_{b'}\}\cap(\{s\}\Delta\{s'\})|.$$
This immediately implies that either $s=s'$ and $b'=b+1$ or $s\neq s'$ and $$b'-b=|\{x_b,\ldots,x_{b'-1},y_{b+1},\ldots,y_{b'}\}\cap \{s,s'\}|\le 2.$$  

In the first case, since $s\in \{x_a,y_a\}$ and $s'\in \{x_{a'},y_{a'}\}$, we must have $a=a'$, and hence $b\neq a-1$, for otherwise $a'=a=b+1=b'$. However, this means that the second condition on $(a,b)$ and $(a',b')$ is satisfied, as desired.

So moving on suppose that the second case holds, i.e. $s\neq s'$, $b'\in \{b+1,b+2\}$ and $b'-b=|\{x_b,\ldots,x_{b'-1},y_{b+1},\ldots,y_{b'}\}\cap \{s,s'\}|$.  

%Similarly, if the monomials of $\mathbf{v}$ and $\mathbf{v}'$ are obtained by omitting \emph{different} variables (i.e., if $a\neq a'$), then the edit-distance in terms of variables between them must be at least four, hence they also cannot be adjacent, a contradiction. So, we must have $b'\in \{b+1,b+2\}$ and $a=a'$ in this case.

Suppose first that $b'=b+1$. Then $\{x_b,y_{b+1}\}$ shares exactly one element with $\{s,s'\}$. Since $s\in \{y_1,\ldots,y_b,x_b,\ldots,x_d\}$ and $s'\in \{y_1,\ldots,y_{b+1},x_{b+1},\ldots,x_d\}$, it follows that either $s=x_b$ or $s'=y_{b+1}$. Recall that $s\in \{x_a,y_a\}$ and $s'\in \{x_{a'},y_{a'}\}$. Hence, in the first case, we obtain $a=b$, a contradiction, and in the second we obtain $a'=b+1=b'$, also a contradiction. It follows that this case is impossible and we may move on with the case $b'=b+2$. We then have that $\{x_b,x_{b+1},y_{b+1},y_{b+2}\}$ shares exactly two common elements with $\{s,s'\}$, i.e., we have that $s,s'$ are distinct elements of $\{x_b,x_{b+1},y_{b+1},y_{b+2}\}$. Since $s\in \{y_1,\ldots,y_b,x_b,\ldots,x_d\}$ and $s'\in \{y_1,\ldots,y_{b+2},x_{b+2},\ldots,x_d\}$, we conclude that in fact $s\in \{x_b,x_{b+1}\}$ and $s'\in \{y_{b+1},y_{b+2}\}$. Recalling further that $a\neq b$ and $a'\neq b'=b+2$, we find that necessarily $s=x_{b+1}$ and $s'=y_{b+1}$. Hence, we have $a=a'=b+1$ and so $b=a-1$. Thus the third of the three conditions for adjacency in $G_d$ is satisfied. This concludes the proof. 
\end{proof}

Note that for any simple polytope $P$, the peak of the silo $S(P,\mathbf{v},\prec)$ is of distance at least $d$ from any of the original vertices in $P$ distinct from $\mathbf{v}$, since its feasible basis is disjoint from all feasible bases of original vertices in $P$, and hence the symmetric difference with these bases is of size $2d$. 

Under the graph isomorphism in Lemma \ref{lem:isomorphism}, the new vertices of the silo whose corresponding monomial has $\mathbf{x}$-support of size $d-1$ are associated to the vertices $(1,2)$ and $(i,1)$ for $2 \leq i \leq d$ of $G_{d}$. These are also exactly the vertices that have some neigbor in the original polytope. At the same time, the new vertices of the silo whose corresponding monomial has $\mathbf{y}$-support of size $d-1$ are associated to the vertices $(i,d)$ for $1 \leq i \leq d-1$ and $(d,d-1)$ of $G_d$. These are also exactly the neighbors of the peak in the silo.

With our next lemma below, we will bound the distances between pairs of new vertices in a silo. To do so, by Lemma~\ref{lem:isomorphism} it suffices to bound the distances between the associated vertices in the $d$-th silo graph $G_d$. Parts (a) and (b) of the following lemma show that any new vertex of a silo $S(P,\mathbf{v},\prec)$ adjacent to an original vertex of $P$ has a path of length at most $d-2$ to a vertex of the silo adjacent to the peak. In fact, part (a) shows that there is a path of length at most $d-2$ from the vertex with associated monomial $\mathbf{x}^{[d]\setminus \{1\}}y_{2}$ to all but one of the $d$ neighbors of the peak in the silo. This flexibility of endpoints of paths starting from the vertex represented by $\mathbf{x}^{[d] \setminus \{1\}}y_{2}$ in conjunction with a rotation action will later allow us to effectively analyze a specific variant of the ``repeated siloing'' construction mentioned further above. Finally, part~(c) of the lemma shows that all new vertices of $S(P,\mathbf{v},\prec)$ can reach a vertex with a neighbor in the original polytope in at most $d-2$ steps, and part~(d) guarantees that any two vertices adjacent to original vertices of $P$ are close to each other.

\begin{lem}
\label{lem:dminus2}
Let $G_{d}$ be the $d$-th silo graph. Then 
\begin{itemize}
    \item[(a)] There is a path of length $d-2$ from $(1,2)$ to $(i,d)$ for all $i \in [d-1] \setminus \{2\}$ and to $(d,d-1)$.
    \item[(b)] For each $2 \leq i \leq d-1$, there is a path of length $d-2$ from $(i,1)$ to $(i,d)$ and from $(d,1)$ to $(d,d-1)$. 
    \item[(c)] Every vertex of $G_d$ can reach some vertex in $\{(1,2)\}\cup \{(i,1)|2\le i \le d\}$ in at most $d-2$ steps.
    \item[(d)] Any two vertices in the set $\{(1,2)\}\cup \{(i,1)|2\le i\le d\}$ have distance at most $3$ from each other.
\end{itemize}
\end{lem}

\begin{proof}
For (a), if $i = 1$, simply increase the second coordinate until reaching $(1,d)$, and this takes at most $d-2$ steps. If $3 \leq i \leq d-1$, take one step to move from $(1,2)$ to $(i,2)$. Then increase the second coordinate until reaching $(i,d)$. This takes at most $d-3$ steps, since each step increases the second index by $1$ except for the step from $(i,i-1)$ to $(i,i+1)$, which increases it by $2$. Thus, it takes $d-2$ steps overall to reach $(i,d)$ for $3 \leq i \leq d-1$. Finally, for moving to $(d,d-1)$, first take $1$ step to move from $(1,2)$ to $(d,2)$ and then increase the second coordinate $d-3$ times to reach $(d,d-1)$ in $d-2$ steps.

For (b), increase the second coordinate iteratively. This takes $d-2$ steps for moving from $(i,1)$ to $(i,d)$ for $i \leq d-1$, because all except one of the steps increase the second coordinate by $1$, and as in the justification for part (a), the second coordinate increases by $2$ from $(i,i-1)$ to $(i,i+1)$. For moving from $(d,1)$ to $(d,d-1)$, increasing the second coordinate straightforwardly takes $d-2$ steps. 

%For (c), let $(c,d)$ be any other vertex. Then if $b \neq c$, move from $(a,b)$ to $(c,b)$ in a single step. Then changing the second coordinate until reaching $(c,d)$ takes at most $d-2$ steps as there are only $d-1$ vertices with the same first coordinate. 

%Suppose instead that $b = c$. Suppose that $a \neq d$. Then move from $(a,b)$ to $(a,d)$ in at most $d-2$ steps and then to $(c,d)$ in $1$ step.

%Thus, the only remaining case is when $c= b$ and $d = a$, so the path moves from $(a,b)$ to $(b,a)$. Suppose without loss of generality that $a \leq b$. Since $a \neq b$, we have $a < b$. Move to $(x,b)$ for some $x \notin \{a,b\}$. This takes $1$ step. Then move from $(x,b)$ to $(x,a)$. We claim that the latter takes at most $d-3$ steps: Since $1\le a<b\le d-1$, it certainly takes at most $b-a$ steps, which is good enough to establish the desired bound unless $a=1, b=d-1$. But in this case it also takes at most $d-3$ steps, since then $a<x<b$ and so we than decrease the second coordinate once by two when Then move from $(x,a)$ to $(b,a)$ in $1$ step. This takes at most $d-3 + 1 + 1 = d-1$ steps. 
For (c), consider any vertex of $(a,b)$ of $G_d$. Note that the vertices of $G_d$ in $\{(x,b)|x\in [d]\setminus \{b\}\}$ induce a path on $d-1$ vertices (and hence of length $d-2$) as a subgraph of $G_d$, whose vertex with the lowest second coordinate is $(2,1)$ for $b=1$ and $(1,b)$ otherwise. Hence, by moving along this path we can always connect $(a,b)$ to a vertex in $\{(1,2)\}\cup \{(i,1)|2\le i \le d\}$ in at most $d-2$ steps, as desired.

For (d), it suffices to note that by definition of the graph $G_d$, the set of vertices $\{(i,1)|2\le i\le d\}$ form a clique, and hence have pairwise distance one. Furthermore, since $(1,2)$ is adjacent to $(3,2)$, which is adjacent to $(3,1)$ in $G_d$, the vertex $(1,2)$ has distance at most $2$ from some vertex in this clique. Hence, it has distance at most $3$ from any vertex in this clique, proving the desired statement.
\end{proof}

The following corollary records some further observations and consequences of Lemma~\ref{lem:isomorphism} and Lemma~\ref{lem:dminus2} in a somewhat different language, which shall become useful later.

\begin{cor}
\label{cor:d-2paths}
Let $\mathbf{v}$ be a vertex of a simple polytope $P$ and let $\prec$ be a linear order on the elements $b_1\prec\cdots\prec b_d$ of the corresponding feasible basis. Let $H$ be the subgraph of the graph of $S(P,\mathbf{v},\prec)$ induced by the vertices not in $P$ and distinct from the peak, and let $\phi:V(G_d)\rightarrow V(H)$ denote the graph isomorphism from Lemma~\ref{lem:isomorphism}. Let $\mathbf{u}_1:=\phi(1,2)$, $\mathbf{u}_i:=\phi(i,1)$ for $2\le i\le d$, $\mathbf{v}_i:=\phi(i,d)$ for $1\le i\le d-1$ and $\mathbf{v}_d:=\phi(d,d-1)$. Furthermore, for $1\le i \le d$ let $\mathbf{s}_i$ denote the unique neighbor of $\mathbf{v}$ on $P$ whose feasible basis contains the elements $\{b_1,\ldots,b_d\}\setminus \{b_i\}$. Then the following hold:

\begin{itemize}
    \item $\mathbf{s}_i$ and $\mathbf{u}_i$ are adjacent on $S(P,\mathbf{v},\prec)$ for every $i\in [d]$.
    \item For each $i\in [d]$, there exists a path of length at most $d-2$ from $\mathbf{u}_i$ to $\mathbf{v}_i$ in $H$, as well as a path of length at most $d-2$ from $\mathbf{u}_1$ to $\mathbf{v}_i$ for each $i\neq 2$.
    \item $\mathbf{v}_1,\ldots,\mathbf{v}_d$ are the neighbors of the peak of $S(P,\mathbf{v},\prec)$.
    \item Every vertex in $H$ has distance at most $d-2$ from some vertex in $\{\mathbf{u}_1,\ldots,\mathbf{u}_d\}$.
    \item Any two of $\mathbf{u}_1,\ldots,\mathbf{u}_d$ have distance at most $3$ from each other in $H$.
\end{itemize}
\end{cor}
\begin{proof}
The first and the third items can easily be checked by inspecting the monomials representing the respective vertices and observing that they only differ by exchanging a single variable.

The second item follows immediately by combining Lemmas~\ref{lem:isomorphism} and~\ref{lem:dminus2}. Finally, the last two items follow directly from parts (c) and (d) of Lemma~\ref{lem:dminus2}, since $\{\mathbf{u}_1,\ldots,\mathbf{u}_d\}$ is the image of $\{(1,2)\}\cup \{(i,1)|2\le i \le d\}$ under the graph isomorphism $\phi$.
\end{proof}

% Given a polytope $P$, a vertex $\mathbf{v}$ of $P$, and a neighbor $\mathbf{u}$ of that vertex, we can add $d$ inequalities  $(\mathbf{u},\mathbf{v})$-silo $Q$, for which the graph of $Q$ is equal to the graph of $P$ with $v$ replaced with a copy of $G_{d}$ and a new vertex $\mathbf{v}'$ such that there is a path from $\mathbf{u}$ to all but one neighbor of $\mathbf{v}'$ of length at most $d-1$. 

Guided by these structural observations about short paths between the new vertices in a silo, we will now introduce the previously announced construction which involves repeated siloing in a cyclic manner, which we call \textbf{cyclic siloing}. In the following, we will repeatedly use the following notation: For an integer $z\in \mathbb{Z}$, we denote by $\overline{z}$ the unique member of $[d]$ which is congruent to $z$ modulo $d$. 

To explain this construction, suppose we are given as input a simple $d$-dimensional polytope $P$ described by $m$ inequalities, and a vertex $\mathbf{v}$. Suppose further we are given as input some positive integer $r$. We will then construct, in polynomial time in the encoding length of $P$ and in $r$, a sequence of simple $d$-dimensional polytopes $C_0,C_1,\ldots,C_{rd}$, and for each $0\le j\le rd$ a special vertex $\mathbf{v}_j$ on $C_j$, an enumeration $\mathbf{v}_{1,j},\ldots,\mathbf{v}_{j,d}$ of the $d$ neighbors of $\mathbf{v}_j$ on $C_j$, and for every $j\in [rd]$ another sequence $\mathbf{u}_{1,j},\ldots,\mathbf{u}_{d,j}$ of special vertices on $C_j$, as follows.

\begin{itemize}
    \item To initialize, for $j=0$ we set $C_0:=P$, $\mathbf{v}_0:=\mathbf{v}$. Finally, we fix some arbitrary enumeration $\mathbf{v}_{1,0},\ldots,\mathbf{v}_{d,0}$ of the $d$ neighbors of $\mathbf{v}$ on $P$.  
    \item Next let $j\ge 1$, and suppose we already computed $C_{j-1},\mathbf{v}_{j-1}$ and enumerated the $d$ neighbors of $\mathbf{v}_{j-1}$ on $C_{j-1}$ as $\mathbf{v}_{1,j-1},\ldots,\mathbf{v}_{d,j-1}$.
    Now compute the unique labeling $b_{1,j-1},\ldots,b_{d,j-1}$ of the elements of the feasible basis of $C_{j-1}$ corresponding to $\mathbf{v}_{j-1}$ such that the tight inequalities shared by $\mathbf{v}_{i,j-1}$ and $\mathbf{v}_{j-1}$ are exactly $\{b_{1,j-1},\ldots,b_{d,j-1}\}\setminus \{b_{i,j-1}\}$. Let $\prec_{j-1}$ be the linear order on the elements of the feasible basis of $\mathbf{v}_{j-1}$ in $C_{j-1}$ defined by $$b_{\overline{j},j-1}\prec_{j-1}b_{\overline{j+1},j-1}\prec_{j-1}\cdots\prec_{j-1}b_{\overline{j+d-1},j-1}.$$ We then set $C_j:=S(C_{j-1},\mathbf{v}_{j-1},\prec_{j-1})$. Next, we set $\mathbf{v}_j$ to be the peak of the silo $C_j$. Finally, we consider the isomorphism $\phi_j$ from the $d$-th silo graph $G_d$ to the subgraph $H_j$ of the graph of $C_j$ induced by all vertices not in $C_{j-1}$ and distinct from the peak $\mathbf{v}_{j}$, as described in Lemma~\ref{lem:isomorphism}. We then set $\mathbf{u}_{i,j}:=\phi_j(\overline{i-j+1},1)$ for all $i\in [d]\setminus \{\overline{j}\}$  and $\mathbf{u}_{\overline{j},j}:=\phi_j(1,2)$. Furthermore, we set $\mathbf{v}_{i,j}:=\phi_j(\overline{i-j+1},d)$ for all $i\in [d]\setminus \{\overline{j-1}\}$ and $\mathbf{v}_{\overline{j-1},j}:=\phi_j(d,d-1)$. Note that by Corollary~\ref{cor:d-2paths}, the so-defined vertices $\mathbf{v}_{1,j},\ldots,\mathbf{v}_{d,j}$ indeed form the neighbors of the peak $\mathbf{v}_j$ of $C_j$.
\end{itemize}

Finally, the last polytope in our construction, $C_{rd}$, is a simple $d$-dimensional polytope, which we call the \textbf{$r$-cyclic siloing} of $P$ and denote by $C^r(P,\mathbf{v})$. Note that by following the steps described above, and by Lemma~\ref{lem:complexity}, given as input $P$, $\mathbf{v}$, $\prec$ and $r$ we can compute in polynomial time in the encoding length of $P$ and in $r$ an inequality description of $C^r(P,\mathbf{v})$ whose encoding length is polynomial in the encoding length of $P$ and in $r$.

We call the subgraph of the graph of the $r$-cyclic siloing of a polytope $P$ induced by all vertices not in the original polytope $P$ the \textbf{cyclic silo}. We also refer to the set of vertices $\{\mathbf{u}_{1,1},\ldots,\mathbf{u}_{d,1}\}$ of the cyclic silo as the \textbf{ground layer} of the cyclic silo.

Before proceeding to analyze the distances between vertices on the $r$-cyclic siloing of a polytope in more detail, we record the following observation which follows directly from our previous lemmas and the construction described above.

\begin{remark}\label{rem:record}
\noindent
\begin{itemize}
    \item For every $(i,j)\in [d]\times[rd]$ we have that $\mathbf{u}_{i,j}$ and $\mathbf{v}_{i,j}$ are vertices in the cyclic silo of $C^r(P,\mathbf{v})$. 
    \item For every $(i,j)\in [d]\times [rd]$, we have that $\mathbf{v}_{i,j-1}$ and $\mathbf{u}_{i,j}$ are adjacent vertices of the polytope $C^r(P,\mathbf{v})$.
    \item For every $(i,j)\in [d]\times [rd]$, we have that $\mathbf{u}_{i,j}$ and $\mathbf{v}_{i,j}$ have distance at most $d-2$ in the cyclic silo. Additionally, if $i\neq \overline{j+1}$, then $\mathbf{u}_{\overline{j},j}$ and $\mathbf{v}_{i,j}$ have distance at most $d-2$ in the cyclic silo.
    \item Any two vertices in the ground layer have distance at most $3$ in the cyclic silo.
\end{itemize}
\end{remark}
\begin{proof}
\noindent
\begin{itemize}
    \item Consider any $(i,j)\in [d]\times [rd]$. Then by definition of the process for building the $r$-cyclic siloing of $P$, we have that $\mathbf{u}_{i,j}$ and $\mathbf{v}_{i,j}$ are vertices of the silo $C_j$ constructed during the process distinct from the peak $\mathbf{v}_j$ of $C_j$. In the process, the final polytope $C^r(P,\mathbf{v})=C_{rd}$ arises from $C_j$ by repeatedly siloing vertices, starting with $\mathbf{v}_j$ and then continuing always by siloing \emph{new} vertices created in the previous siloing step. In particular, throughout the rest of the process, no vertex of $C_j$ except $\mathbf{v}_j$ gets siloed, and hence all vertices of $C_j$ distinct from $\mathbf{v}_j$ remain vertices of $C^r(P,\mathbf{v})$. This includes $\mathbf{u}_{i,j}$ and $\mathbf{v}_{i,j}$, confirming the statement claimed in the first item. 
    \item By the same argument we can observe that for every $j\in \{0,\ldots,rd\}$ all adjacencies between vertices of $C_j$ distinct from $\mathbf{v}_j$ remain intact in the final polytope $C^r(P,\mathbf{v})$. Hence, for the second item it suffices to verify that $\mathbf{v}_{i,j-1}$ and $\mathbf{u}_{i,j}$ are adjacent on the polytope $C_j=S(C_{j-1},\mathbf{v}_{j-1},\prec_{j-1})$. By our description of the process, we have that $\mathbf{v}_{i,j-1}$ is the unique neighbor of $\mathbf{v}_{j-1}$ on $C_{j-1}$ whose feasible basis includes $\{b_{1,j-1},\ldots,b_{d,j-1}\}\setminus \{b_{i,j-1}\}$. By definition of $\prec_{j-1}$, this means that $\mathbf{v}_{i,j-1}$ corresponds to the vertex $\mathbf{s}_{\overline{i-j+1}}$ in the notation of Corollary~\ref{cor:d-2paths} (applied to $C_{j-1}$ and $\prec_{j-1}$ instead of $P$ and $\prec$). Hence, it follows from the first item of Corollary~\ref{cor:d-2paths} that $\mathbf{v}_{i,j-1}$ is adjacent to $\phi_j(\overline{i-j+1},1)$ on $C_j$ if $i\neq \overline{j}$ and to $\phi_j(1,2)$ on $C_j$ if $i=\overline{j}$. By definition of $\mathbf{u}_{i,j}$ in the process, it follows that indeed $\mathbf{v}_{i,j-1}$ is adjacent to $\mathbf{u}_{i,j}$ on $C_j$, as desired.
    \item Let $(i,j)\in [d]\times [rd]$. Recall that $C_{j}=S(C_{j-1},\mathbf{v}_{j-1},\prec_{j-1})$, where $\prec_j$ is the linear order on the elements of the feasible basis representing $\mathbf{v}_{j-1}$ in $C_{j-1}$, defined as 
$$b_{\overline{j},j-1}\prec_{j-1} b_{\overline{j+1},j-1}\prec_{j-1} \cdots \prec_{j-1} b_{\overline{j+d-1},j-1}.$$ Recall that $\phi_{j}$ denotes the isomorphism from $G_d$ to $H_{j}$ as given by Lemma~\ref{lem:isomorphism}, and that $\mathbf{u}_{i,j}=\phi_{j}(\overline{i-j+1},1)$ for $i\in [d]\setminus\{\overline{j}\}$ as well as $\mathbf{u}_{\overline{j},j}=\phi_{j}(1,2)$ by definition in the process. Similarly, we have $\mathbf{v}_{i,j}=\phi_{j}(\overline{i-j+1},d)$ for all $i\in [d]\setminus \{\overline{j-1}\}$ and $\mathbf{v}_{\overline{j-1},j}=\phi_{j}(d,d-1)$.  

Hence, for each $i\in [d]$, the vertices $\mathbf{u}_{i,j}$ and $\mathbf{v}_{i,j}$ correspond exactly to the vertices $\mathbf{u}_{\overline{i-j+1}}$ and $\mathbf{v}_{\overline{i-j+1}}$ in the notation of Corollary~\ref{cor:d-2paths}, applied to $C_{j}=S(C_{j-1},\mathbf{v}_{j-1},\prec_{j-1})$ in place of $S(P,\mathbf{v},\prec)$. It thus follows by the second item of Corollary~\ref{cor:d-2paths} that for each $i\in [d]$ there exists a path of length at most $d-2$ on $C_{j}$ from $\mathbf{u}_{i,j}$ to $\mathbf{v}_{i,j}$, and from $\mathbf{u}_{\overline{j},j}$ to $\mathbf{v}_{i,j}$, provided that $\overline{i-j+1}\neq 2$, i.e., $i\neq \overline{j+1}$. Moreover, all vertices of these paths are distinct from the peak of $C_{j}$. As we argued in the first and second item of this proof, this means that these paths also exist in the cyclic silo of $C^r(P,\mathbf{v})$. This establishes the statement claimed in the third item of the remark and concludes its proof.
\item For the fourth statement claimed in the remark, note that in the language of Corollary~\ref{cor:d-2paths} $\mathbf{u}_{1,1},\ldots,\mathbf{u}_{1,d}$ correspond exactly to the vertices $\mathbf{u}_1,\ldots,\mathbf{u}_d$ in the first silo $C_1=S(C_0,\mathbf{v}_0,\prec_0)$ constructed in the process. Hence, by the fifth item of that corollary, we have that any two of $\mathbf{u}_{1,1},\ldots,\mathbf{u}_{1,d}$ can be connected by a path of length at most $3$ in the graph of $C_1$ which does not use the peak $\mathbf{v}_1$. By what we observed before, any such path also forms a path in the cyclic silo of $C^r(P,\mathbf{v})$, and so we obtain the claimed statement.
\end{itemize}
\end{proof}
Our next lemma is a key technical step towards the proof of our Theorem~\ref{thm:combodiam}. It gives an upper bound on the distances between certain pairs of vertices in the cyclic silo.
\begin{lem}
\label{lem:silodiameter}
Let $P$ be a simple $d$-dimensional polytope, $\mathbf{v}$ a vertex of $P$ and let $r\ge 3$ be an integer. Then for every pair of vertices $\mathbf{s},\mathbf{t}$ of the the cyclic silo, at least one of which lies in the ground layer, we have that their distance in the cyclic silo is at most $rd(d-1)$. In particular, the diameter of the cyclic silo is at most $2rd(d-1)$. Furthermore, the distance on $C^r(P,\mathbf{v})$ from the peak $\mathbf{v}_{rd}$ to any vertex of $P$ distinct from $\mathbf{v}$ is at least $rd(d-1)+1$.
\end{lem}
\begin{proof}
%There are $rd$ vertices that get siloed by construction. For $i\in [d]$ and $j\in [rd]$ let $u_{i,j}$ denote the $i$th vertex corresponding to an edge incident to the $j$th vertex that gets siloed. Then, by Corollary \ref{cor:d-2paths} there is a path from $u_{1,1}$ to $u_{i,2}$ of length $d-1$ for each $i \neq 2$ and from $u_{i,1}$ to $u_{i,2}$ of length $d-1$ for all $i \neq 1$. More generally, by the construction of cyclic siloing, pause to note that if $i\equiv j \text{ (mod }d)$, then $u_{i,j}$ will have a path of length $d-1$ to $u_{k,j+1}$ whenever $k \not\equiv i+1 \text{ (mod }d)$, and that $u_{i,j}$ will always have a path of length $d-1$ to $u_{i,j+1}$.

% What I wanted to say was: 
Throughout this proof, we will use the same notation for vertices and polytopes as defined in the description of the process for constructing the $r$-cyclic siloing $C^r(P,\mathbf{v})$. 

In the following, let us define an auxiliary graph $\Gamma$ on vertex-set $\{\mathbf{u}_{i,j}|(i,j)\in [d]\times [rd]\}$, where we make $\mathbf{u}_{i,j}$ and $\mathbf{u}_{i',j'}$ adjacent if and only if there is a path of length at most $d-1$ between them in the cyclic silo. It follows from the first three items of Remark~\ref{rem:record} that there is an edge in $\Gamma$ from $\mathbf{u}_{i,j}$ to $\mathbf{u}_{i,j+1}$ for all $(i,j)\in [d]\times [rd-1]$. Additionally, Remark~\ref{rem:record} implies that $\mathbf{u}_{\overline{j},j}$ is adjacent to $\mathbf{u}_{i,j+1}$ in $\Gamma$ provided $i \neq \overline{j+1}$. It will be useful to first prove the following claim.

\textbf{Claim~1.} Let $i,i'\in [d]$ and $j\in [rd]$ be such that $j\ge d+3$. Then $\mathbf{u}_{i,1}$ and $\mathbf{u}_{i',j}$ have distance at most $(d-1)(j-1)$ in the cyclic silo.
\begin{claimproof}
By definition of the graph $\Gamma$, it clearly suffices to show that the distance between $\mathbf{u}_{i,1}$ and $\mathbf{u}_{i',j}$ in $\Gamma$ is at most $j-1$. To do so, we will construct a path from $\mathbf{u}_{i,1}$ to $\mathbf{u}_{i',j}$ in $\Gamma$ where whenever we move along an edge of the path, the second index of the current vertex increases by exactly $1$. Clearly, such a path will always have the desired length of $j-1$.

Suppose first that $i'\neq \overline{i+1}$.
Starting at $\mathbf{u}_{i,1}$, we can then move in $\Gamma$ to $\mathbf{u}_{i,i}$ in $i-1$ steps by successively increasing the second index. Next, we can move to $\mathbf{u}_{i',i+1}$ in one step. Finally, we again increase the second index successively until arriving at the desired vertex $\mathbf{u}_{i',j}$ along a path of the desired in type in $\Gamma$. This is possible since $i+1\le d+1\le j$.

Next, suppose $i'=\overline{i+1}$. As before, starting at $\mathbf{u}_{i,1}$ we first successively increase the second index to reach $\mathbf{u}_{i,i}$. We then move from $\mathbf{u}_{i,i}$ to $\mathbf{u}_{\overline{i+2},i+1}$ in one step and then to $\mathbf{u}_{\overline{i+2},i+2}$ in another step. Since $\overline{i+1}\neq \overline{i+3}$ (here we use $d\ge 3$), we can next move to $\mathbf{u}_{\overline{i+1}, i+3}$. Finally we successively increase the second coordinate until reaching $\mathbf{u}_{\overline{i+1},j}=\mathbf{u}_{i',j}$ via a path in $\Gamma$ of the desired form. This is possible since $i+3\le d+3\le j$.

Having found the desired path in $\Gamma$ of length $j-1$ in both cases, we may conclude the proof.
\end{claimproof}
Next we want to show the desired upper bound on the distance between pairs $\mathbf{s},\mathbf{t}$ of vertices in the cyclic silo at least one of which belongs to the ground layer.

We first prepare some useful setup. Recall that for $j\in [rd]$ we denote by $H_j$ the subgraph of the graph of $C_j$ induced by all vertices distinct from the peak $\mathbf{v}_j$ of $C_j$ that are not vertices of $C_{j-1}$. Pause to note that each $H_j$ is in fact a subgraph of the cyclic silo, and that the vertices of the cyclic silo partition into the disjoint sets of vertices $V_1:=V(H_1),\ldots,V_{rd}:=V(H_{rd})$ and $V_{rd+1}:=\{\mathbf{v}_{rd}\}$. Furthermore, note that each graph $H_j$ is isomorphic to the $d$-th silo graph by Lemma~\ref{lem:isomorphism}. Additionally, it follows straightforwardly from the fourth item of Corollary~\ref{cor:d-2paths} and from the definition of the vertices $\mathbf{u}_{i,j}$ in the process of constructing the $r$-cyclic siloing that for each $j\in [rd]$, every vertex $\mathbf{w}$ in $H_j$ has distance at most $d-2$ from \emph{some} vertex in $\{\mathbf{u}_{1,j},\ldots,\mathbf{u}_{d,j}\}$.

So let now a pair $\mathbf{s},\mathbf{t}$ of vertices of the cyclic silo be given to us such that $\mathbf{s}$ belongs to the ground layer. Let $i\in [d]$ be such that $\mathbf{s}=\mathbf{u}_{i,1}$ and let $j\in [rd+1]$ be the unique index such that $\mathbf{t}\in V_{j}$.

Suppose first that $j\le rd$. Then $\mathbf{t}\in V(H_j)$. By our above remark, there then exists some $i'\in [d]$ such that $\mathbf{t}$ has distance at most $d-2$ from $\mathbf{u}_{i',j}$ in $H_j$, and hence in the cyclic silo.

If $j\ge d+3$, then by Claim~1 we have that $\mathbf{s}=\mathbf{u}_{i,1}$ and $\mathbf{u}_{i',j}$ have distance at most $(d-1)(j-1)$ in the cyclic silo. By the triangle inequality, we then find that $\mathbf{s}$ and $\mathbf{t}$ have distance at most $(d-1)(j-1)+d-2=j(d-1)-1< rd(d-1)$ in the cyclic silo, as desired. 

On the other hand, if $j\le d+2$, then by successively decreasing the second index we can see that $\mathbf{u}_{i',j}$ has a path of length at most $j-1$ in $\Gamma$ to $\mathbf{u}_{i',1}$. In particular, the distance between $\mathbf{u}_{i',j}$ and $\mathbf{u}_{i',1}$ in the cyclic silo is at most $(d-1)(j-1)$. Furthermore, by the last item of Remark~\ref{rem:record}, we have that $\mathbf{s}=\mathbf{u}_{i,1}$ and $\mathbf{u}_{i',1}$ have distance at most $3$ in the cyclic silo. Altogether, it follows from the triangle inequality that $\mathbf{s}$ and $\mathbf{t}$ have distance at most $3+(d-1)(j-1)+(d-1)=j(d-1)+3\le (d+2)(d-1)+3\le rd(d-1)$, where we used our assumption $r\ge 3$ as well as our general assumption that $d\ge 3$, in the last step. 

Hence, it only remains to consider the case that $j=rd+1$, i.e., $\mathbf{t}=\mathbf{v}_{rd}$. Then, since by definition $\mathbf{v}_{rd}$ is the peak of the silo $C_{rd}=S(C_{rd-1},\mathbf{v}_{rd-1},\prec_{rd-1})$, it follows by the third item of Corollary~\ref{cor:d-2paths} that each of the vertices $\mathbf{v}_{rd,1},\ldots,\mathbf{v}_{rd,d}$ are the neighbors of $\mathbf{t}=\mathbf{v}_{rd}$ on $C_{rd}$ and hence in the cyclic silo. Further, by the third item of Remark~\ref{rem:record} we have that $\mathbf{u}_{i,rd}$ has distance at most $d-2$ from $\mathbf{v}_{i,rd}$ in the cyclic silo, and hence distance at most $d-1$ from $\mathbf{v}_{rd}=\mathbf{t}$ in the cyclic silo. By Claim~1 or simply by walking in the graph $\Gamma$, we can also see that $\mathbf{s}=\mathbf{u}_{i,1}$ has distance at most $(d-1)(rd-1)$ from $\mathbf{u}_{i,1}=\mathbf{s}$ in the cyclic silo. Hence, by the triangle inequality, we find that $\mathbf{s}$ and $\mathbf{t}$ have distance at most $(d-1)(rd-1)+(d-1)=rd(d-1)$ in the cyclic silo. This concludes the proof of the first part of the lemma.  

It remains to prove that the distance on $C_{rd}=C^r(P,\mathbf{v})$ from the peak $\mathbf{v}_{rd}$ to any vertex of $P$ distinct from $\mathbf{v}$ is at least $rd(d-1)+1$. We will do this by proving the following, more general, statement by induction.

\textbf{Claim~2.} For every $j\in \{0,1,\ldots,rd\}$ and every $i\in [d]$, the distance on $C_j$ between $\mathbf{v}_j$ and any vertex of $P$ distinct from $\mathbf{v}$ is at least $j(d-1)+1$.
\begin{claimproof}
The induction basis $j=0$ is obvious, since $C_0=P$ and $\mathbf{v}_0=\mathbf{v}$. So suppose that for some $j\in [rd]$ we already proved that the distance on $C_{j-1}$ from $\mathbf{v}_{j-1}$ to any vertex of $P$ distinct from $\mathbf{v}$ is at least $(j-1)(d-1)+1$, and let us prove the analogous claim with $j-1$ replaced by~$j$. Let $R$ denote a shortest path from $\mathbf{v}_j$ to some vertex $\mathbf{w}$ of $P$ distinct from $\mathbf{v}$ on $C_j$. Recall that $C_j=S(C_{j-1},\mathbf{v}_{j-1},\prec_{j-1})$, and pause to note that by definition of the siloing operation, the members of the set $N:=\{\mathbf{v}_{1,j-1},\ldots,\mathbf{v}_{d,j-1}\}$ are the only vertices of $C_{j-1}$ who have a neighbor on $C_j$ not in $C_{j-1}$. Hence, $N$ forms a separator between  $\mathbf{w}$ and $\mathbf{v}_j$ in the graph of $C_j$. In particular, $R$ must contain at least one vertex in $N$. Let $\mathbf{n}$ denote the vertex in $N$ which we meet first when traversing $R$ from $\mathbf{w}$ to $\mathbf{v}_j$. Note that all vertices of the segment of $R$ from $\mathbf{w}$ to $\mathbf{n}$ must be vertices of $C_{j-1}$, for otherwise there would be a vertex in $N$ along the path that is closer to $\mathbf{w}$ than $\mathbf{n}$. Hence, we find that the length of the segment of $R$ from $\mathbf{w}$ to $\mathbf{n}$ is at least $d_{C_{j-1}}(\mathbf{w},\mathbf{n})\ge d_{C_{j-1}}(\mathbf{w},\mathbf{v}_{j-1})-d_{C_{j-1}}(\mathbf{n},\mathbf{v}_{j-1})=d_{C_{j-1}}(\mathbf{w},\mathbf{v}_{j-1})-1\ge (j-1)(d-1)$, where we used the fact that $\mathbf{n}\in N$ is a neighbor of $\mathbf{v}_{j-1}$ on $C_{j-1}$ in the second to last step, and the inductive assumption in the last step. Now, let us consider the segment of $R$ from $\mathbf{n}$ to $\mathbf{v_j}$. Since $\mathbf{n}$ is a vertex of $C_{j-1}$ and $\mathbf{v}_j$ the peak of a siloing performed on $C_{j-1}$, we have that the feasible basis representing $\mathbf{v}_j$ on $C_j$ is disjoint from that of $\mathbf{n}$ (compare also our analogous remark directly after the proof of Lemma~\ref{lem:isomorphism}). Hence, the distance between $\mathbf{n}$ and $\mathbf{v}_j$ on $C_j$ and thus the length of the segment of $R$ from $\mathbf{n}$ to 
$\mathbf{v}_j$ must be at least $d$. It now follows that $d_{C_j}(\mathbf{w},\mathbf{v}_j)=|R|\ge (j-1)(d-1)+d=j(d-1)+1$, as desired. This established the inductive claim and concludes the proof.
\end{claimproof}
\end{proof}

With Lemma~\ref{lem:silodiameter}, the key technical point of our argument has now been established. The construction we use for our hardness reduction from \textsc{$k$-Distance on Simple Polytopes} to \textsc{Diameter of simple polytopes} will look as follows. We are given as input a $d$-dimensional simple polytope $P$ and the two vertices $\mathbf{u}, \mathbf{v}$ of a simple input polytope $P$ for which we want to decide if their distance on $P$ is at most some input number $k\in \mathbb{N}$. We will then apply, for some large enough choice of $r\in \mathbb{N}$, the $r$-cyclic siloing operation both at $\mathbf{u}$ and then at $\mathbf{v}$. Finally, our goal in the following will be to show that the diameter of the arising simple polytope $Q$ equals $d_P(\mathbf{u},\mathbf{v})+2rd(d-1)$, and hence we will be able to determine the distance between $\mathbf{u}$ and $\mathbf{v}$ on $P$ by computing the diameter of $Q$, establishing the desired reduction.

As a first step we prove the following lemma, bounding the effect that $r$-cyclic siloing can have on the distances between the original vertices of the polytope we apply it to.

\begin{lem}
\label{lem:nonsiloedvertices}
Let $P$ be a simple polytope and $\mathbf{v}$ a vertex of $P$. Let $\mathbf{s}, \mathbf{t}$ be vertices of $C^r(P,\mathbf{v})$ not contained in the cyclic silo. In particular, $\mathbf{s}$ and $\mathbf{t}$ are also vertices of $P$. Then 
\[d_{C^r(P,\mathbf{v})}(\mathbf{s},\mathbf{t}) \leq d_{P}(\mathbf{s},\mathbf{t}) + 3.\]
\end{lem}

\begin{proof}
Suppose first that some shortest path from $\mathbf{s}$ to $\mathbf{t}$ on $P$ does not pass through $\mathbf{v}$. Then it is still a path in $C^r(P,\mathbf{v})$, so $d_{C^r(P,\mathbf{v})}(\mathbf{s},\mathbf{t}) \le d_{P}(\mathbf{s},\mathbf{t})$.

Next suppose that every shortest path between $\mathbf{s}$ and $\mathbf{t}$ on $P$ does go through $\mathbf{v}$. Pick one such shortest path $R$, and note that it uses precisely two edges incident to $\mathbf{v}$. Recall that in the process for constructing the $r$-cyclic siloing $C^r(P,\mathbf{v})$ of $P$ at $\mathbf{v}$, the vertices $\mathbf{v}_{1,0},\ldots,\mathbf{v}_{d,0}$ denote the neighbors of $\mathbf{v}$ on $P$. Hence, $R$ must use vertices $\mathbf{v}_{i,0},\mathbf{v},\mathbf{v}_{i',0}$ in this order, for some distinct $i,j\in [d]$. By the second item of Remark~\ref{rem:record}, we have that on the polytope $C^r(P,\mathbf{v})$, the vertices $\mathbf{v}_{i,0}, \mathbf{u}_{i,1}$ and $\mathbf{v}_{i',0}, \mathbf{u}_{i',1}$ are adjacent. Furthermore, by the fourth item of Remark~\ref{rem:record}, we have that there exists a path $R'$ on $C^r(P,\mathbf{v})$ between $\mathbf{u}_{i,1}$ and $\mathbf{u}_{i',1}$ of length at most $3$, all whose vertices are in the cyclic silo of $C^r(P,\mathbf{v})$ (and hence, in particular $R'$ shares no vertices with $R$). We can now see that by replacing the subpath $\mathbf{v}_{i,0},\mathbf{v},\mathbf{v}_{i',0}$ of $R$ of length two by the path of length at most $5$ in $C^r(P,\mathbf{v})$ obtained as the union of the edges $\mathbf{v}_{i,0}\mathbf{u}_{i,1}$, $\mathbf{v}_{i',0}\mathbf{u}_{i',1}$ and the path $R'$ defined above, we obtain a path between $\mathbf{s}$ and $\mathbf{t}$ on $C^r(P,\mathbf{v})$. Hence, we have $d_{C^r(P,\mathbf{v})}(\mathbf{s},\mathbf{t})\le d_P(\mathbf{s},\mathbf{t})-2+5=d_P(\mathbf{s},\mathbf{t})+3$, as desired. This concludes the proof.
\end{proof}

With this auxiliary statement at hand, we are now finally in the position to prove the desired formula for the diameter of the polytope $Q$ obtained after $r$-cyclic siloing at two given vertices $\mathbf{u}, \mathbf{v}$ of a polytope $P$, as we mentioned above.

\begin{theorem}\label{thm:diameterprecise}
Let $P$ be a simple polytope with vertices $\mathbf{u}\neq \mathbf{v}$. Let $Q$ be the result of applying an $r$-cyclic siloing at both $u$ and $v$, where $r \geq \max(\mathrm{diam}(P),6)$. Formally, $Q:=C^r(C^r(P,\mathbf{u}),\mathbf{v})$. Then
\[\mathrm{diam}(Q) = d_{P}(\mathbf{u},\mathbf{v}) + 2rd(d-1).\]
\end{theorem}

\begin{proof}
In the following, let us denote by $W_\mathbf{u}, W_\mathbf{v}$ the sets of vertices in the cyclic silos corresponding to $\mathbf{u}$ and $\mathbf{v}$, respectively. 

Let us first show that $\mathrm{diam}(Q)\le d_P(\mathbf{u},\mathbf{v})+2rd(d-1)$, i.e. that every given pair $\mathbf{s},\mathbf{t}$ of vertices of $Q$ satisfies $d_Q(\mathbf{s},\mathbf{t})\le d_P(\mathbf{u},\mathbf{v})+2rd(d-1)$.

Suppose first that none of $\mathbf{s},\mathbf{t}$ lie in $W_\mathbf{u}\cup W_\mathbf{v}$. Then by Lemma~\ref{lem:nonsiloedvertices}, we have
\[d_{Q}(\mathbf{s},\mathbf{t}) \leq d_{P}(\mathbf{s},\mathbf{t}) + 6 \leq \mathrm{diam}(P) +6 \leq 2r \leq 2rd(d-1)\le d_P(\mathbf{u},\mathbf{v})+2rd(d-1),\] as desired.

Next, suppose that exactly one of $\mathbf{s},\mathbf{t}$ lie in $W_{\mathbf{u}}\cup W_{\mathbf{v}}$. W.l.o.g. (possibly after relabeling), we may then assume $\mathbf{s}\notin W_\mathbf{u}\cup W_{\mathbf{v}}$, $\mathbf{t}\in W_\mathbf{v}$. 

By Lemma~\ref{lem:nonsiloedvertices}, we then have that $d_{C^r(P,\mathbf{u})}(\mathbf{s},\mathbf{v})\le d_P(\mathbf{s},\mathbf{v})+3\le \mathrm{diam}(P)+3$. In particular, this implies that there exists a path $R$ (not traversing $\mathbf{v}$) of length at most $\mathrm{diam}(P)+2$ from $\mathbf{s}$ to a neighbor $\mathbf{v}'$ of $\mathbf{v}$ on $C^r(P,\mathbf{u})$. By the second item of Remark~\ref{rem:record}, applied with $j=1$, it now follows that $\mathbf{v}'$ has a neighbor in the ground layer of $W_{\mathbf{v}}$. Since also $\mathbf{t}\in W_{\mathbf{v}}$, Lemma~\ref{lem:silodiameter} implies that there exists a path of length at most $rd(d-1)$ between this neighbor of $\mathbf{v'}$ and $\mathbf{t}$ on $Q$, all whose vertices lie in $W_\mathbf{v}$. It now follows that $d_Q(\mathbf{s},t)\le |R|+1+rd(d-1)\le \mathrm{diam}(P)+3+rd(d-1)\le 2r+rd(d-1)\le 2rd(d-1)\le d_P(\mathbf{u},\mathbf{v})+2rd(d-1)$, as desired.

Finally, suppose that both of $\mathbf{s},\mathbf{t}$ lie in $W_\mathbf{u}\cup W_\mathbf{v}$. If both lie in the same cyclic silo, then by Lemma~\ref{lem:silodiameter} we have $d_Q(\mathbf{u},\mathbf{v})\le 2rd(d-1)\le d_P(\mathbf{u},\mathbf{v})+2rd(d-1)$, as desired. Hence, moving on we may assume that they lie in different cyclic silos, so we may w.l.o.g. assume $\mathbf{s}\in W_\mathbf{u}, \mathbf{t}\in W_\mathbf{v}$. Now, let $R$ be a path of length $d_P(\mathbf{u},\mathbf{v})-2$ on $P$ connecting a neighbor $\mathbf{u'}$ of $\mathbf{u}$ on $P$ to a neighbor $\mathbf{v}'$ of $\mathbf{v}$ on $P$, and note that all vertices of $R$ are distinct from $\mathbf{u}$ and $\mathbf{v}$. In particular, $R$ is also a path on $Q$. Again by the second item of Remark~3.8, applied with $j=1$, we have that $\mathbf{u}'$ is adjacent on $Q$ to a vertex $\mathbf{u}''$ in the ground layer of $W_\mathbf{u}$, and that $\mathbf{v}'$ is adjacent on $Q$ to a vertex $\mathbf{v}''$ in the ground layer of $W_\mathbf{v}$. Since $\mathbf{s}\in W_\mathbf{u}$ and $\mathbf{t}\in W_\mathbf{v}$, Lemma~\ref{lem:silodiameter} implies that $\mathbf{u}''$ and $\mathbf{s}$ as well as $\mathbf{v}''$ and $\mathbf{t}$ have distance at most $rd(d-1)$ on $Q$. By the triangle inequality, we may thus conclude that $d_Q(\mathbf{s},\mathbf{t})\le rd(d-1)+1+|R|+1+rd(d-1)\le d_P(\mathbf{u},\mathbf{v})+2rd(d-1)$, as desired. Having covered all the cases, this concludes the proof of the upper bound $\mathrm{diam}(Q)\le d_P(\mathbf{u},\mathbf{v})+2rd(d-1)$.

For the lower bound, let us denote by $\mathbf{p}_\mathbf{u}, \mathbf{p}_\mathbf{v}$ the final peaks of the cyclic silos $W_\mathbf{u}, W_\mathbf{v}$, respectively. We will show that $d_Q(\mathbf{p}_\mathbf{u},\mathbf{p}_\mathbf{v})\ge d_P(\mathbf{u},\mathbf{v})+2rd(d-1)$, which will clearly establish the desired lower bound on $\mathrm{diam}(Q)$. By construction of the $r$-cyclic siloing, we have that the set $N_\mathbf{u}$ of neighbors of $\mathbf{u}$ on $P$ separates $W_\mathbf{u}$ from the rest of the graph of $Q$, and similarly the set $N_\mathbf{v}$ of neighbors of $\mathbf{v}$ on $P$ separates $W_\mathbf{v}$ from the rest of the graph of $Q$. Now consider a shortest path from $\mathbf{p}_\mathbf{u}$ to $\mathbf{p}_\mathbf{v}$ on $Q$. Let $\mathbf{u'}$ denote the last vertex of the path in $N_\mathbf{u}$ when traversing it from $\mathbf{p}_\mathbf{u}$ to $\mathbf{p}_\mathbf{v}$. Furthermore, let $\mathbf{v}'$ denote the first vertex of $N_\mathbf{v}$ we meet when traversing the segment of the path from $\mathbf{u}'$ to $\mathbf{p}_\mathbf{v}$. Note that all vertices in the segment of the path between $\mathbf{u}'$ and $\mathbf{v}'$ must also be vertices of $P$, and hence this segment forms a path in $P$ and has length at least $d_P(\mathbf{u}',\mathbf{v}')\ge d_P(\mathbf{u},\mathbf{v})-d_P(\mathbf{u}',\mathbf{u})-d_P(\mathbf{v}',\mathbf{v})=d_P(\mathbf{u},\mathbf{v})-2$.

Furthermore, it follows directly from Lemma~\ref{lem:silodiameter} that the segment of the path from $\mathbf{p}_\mathbf{u}$ to $\mathbf{u}'\in N_\mathbf{u}$, as well as the segment of the path from $\mathbf{p}_\mathbf{v}$ to $\mathbf{v}'\in N_\mathbf{v}$, both must have length at least $rd(d-1)+1$.

Altogether, this implies that
the total length $d_Q(\mathbf{p}_\mathbf{u},\mathbf{p}_\mathbf{v}$) of the shortest path we considered is at least $(rd(d-1)+1)+(d_P(\mathbf{u},\mathbf{v})-2)+(rd(d-1)+1)=d_P(\mathbf{u},\mathbf{v})+2rd(d-1)$. This is what we wanted to prove, and hence we may conclude the proof of the theorem.
\end{proof}

The following result summarizes our observations made and the auxiliary results we proved so far.

\begin{theorem}
\label{thm:Reduction}
Given as input a simple polytope $P$ in inequality description, a pair of vertices $\mathbf{u}, \mathbf{v}$ of $P$ and a number $r\in \mathbb{N}$ such that $r\ge \max(\mathrm{diam}(P),6)$, one can compute another~simple polytope $Q$ (also in inequality description) as well as a constant $K$ for which the diameter of $Q$ equals $d_P(\mathbf{u}, \mathbf{v}) + K$, in time bounded polynomially in the encoding length of $P$ and in $r$. 
\end{theorem}
\begin{proof}
This follows directly by combining Lemma~\ref{lem:complexity} and Theorem~\ref{thm:diameterprecise}.
\end{proof}

There is one last issue to consider before proving our main theorem: Since the polynomial Hirsch conjecture remains open, a priori the combinatorial diameter of the input polytope $P$ for \textsc{$k$-Distance on Simple Polytopes} in our attempted hardness reduction may not be polynomially bounded, potentially rendering the time needed to construct the polytope $Q$ in Theorem~\ref{thm:Reduction} superpolynomial. Hence, for the desired reduction to be polynomial, we must ensure that \textsc{$k$-Distance on Simple Polytopes} when restricted to instances $P$ with polynomially bounded diameter still remains \NP-hard. However, this directly follows from the following lemma, showing that the fractional knapsack polytopes we used in our hardness proof for \textsc{$k$-Distance on Simple Polytopes} in the first half of this paper do in fact have linear diameter.
\begin{lem}
\label{lem:knapsackdiameter}
For any choice of $\mathbf{b} \in \mathbb{Z}^{d}_{>0}$, $P_{\mathbf{b}}$ has combinatorial diameter at most $2(d+2)$.
\end{lem}

\begin{proof}
We will show that every vertex in the graph of $P_\mathbf{b}$ has distance at most $d+2$ to the vertex $\emptyset$, which will clearly imply the desired bound on the diameter.

Let us first consider a vertex of the form $(S,i)$ where $S\subseteq [d+2]$ and $i\in [d+2]\setminus S$. Suppose first that $\sum_{i \in S} w_{i} \leq \beta + 1/4$. Then $S$ is a neighbor of $(S,i)$ via an edge of type $(b)$ in Lemma \ref{lem:edgetypes}. Removing all elements of $S$ other than $d+1$ and then finally $d+1$ will lead to the vertex $\emptyset$ after at most $|S| + 1\le (d+1)+1=d+2$ many steps using type (a) moves.

Suppose instead that $\sum_{i \in S} w_{i} \geq \beta + 1/4$. Then move to $S \cup \{i\}$ via a type (b) move and apply the same argument. Notice that since $[d+2]$ is not a vertex, we have $|S|+1=|S\cup \{i\}|\le d+1$ in this situation. Hence, it takes a total of at most $|S| + 2\le (d+1)+1=d+2$ many steps to reach $\emptyset$ from $(S,i)$, as desired. 

For any vertex of the form $T\subseteq [d+2]$, the same decrementing procedure works and takes at most $|T|\le d+2$ many steps to reach vertex $\emptyset$. Since $|S| \leq d+2$ and $|T| \leq d+2$. This concludes the proof.
\end{proof}

\begin{cor}\label{cor:obvious}
    \textsc{$k$-Distance on simple polytopes}, restricted to input instances $P$ of diameter at most $2\dim(P)+4$, is \NP-hard.
\end{cor}
\begin{proof}
    This follows directly from Lemma~\ref{lem:knapsackdiameter} and from our proof of Theorem~\ref{thm:shortestpathsimple}.
\end{proof}

Applying this theorem yields the proof of our second main result:

\begin{proof}[Proof of Theorem \ref{thm:combodiam}]
By Corollary~\ref{cor:obvious}, \textsc{$k$-Distance on Simple Polytopes} restricted to input instances $P$ with diameter at most $2\mathrm{dim}(P)+4$ is \NP-hard. Given an input instance $P,\mathbf{u},\mathbf{v}, k$ of this problem, we then compute $r:=\max(2\mathrm{dim}(P)+4,6)$ and apply Theorem~\ref{thm:Reduction}, using which we can construct, in polynomial time in the encoding length of $P$ and in $r$ (and hence simply in polynomial time in the encoding length of $P$) a simple polytope $Q$ in inequality description and a number $K$ such that $d_P(\mathbf{u},\mathbf{v})=\mathrm{diam}(Q)-K$. Hence, given access to an oracle for \textsc{Diameter of simple polytopes} we can compute the distance between $\mathbf{u}$ and $\mathbf{v}$ on $P$ and hence decide whether $d_P(\mathbf{u},\mathbf{v})\le k$, in polynomial time in the encoding length of $P$ and $k$ plus the time needed for executing the oracle. Hence, we have found a Turing reduction from \textsc{$k$-Distance on Simple Polytopes} restricted to input instances $P$ with diameter at most $2\mathrm{dim}(P)+4$ to \textsc{Diameter of simple polytopes}. Since the former is \NP-hard, so is the latter, concluding the proof.
\end{proof}
\section{Rock Extensions}

In \cite{RockExtensions}, Kaibel and Kukharenko made the stunning observation that linear programming may be reduced in strongly polynomial time to the case of linear programs over a special family of simple polytopes called \textbf{rock extensions}, which have linear diameters. In the degenerate setting, this is trivial as one can simply take a pyramid over the original polytope, and the resulting polytope will have diameter $2$. Hence, the notable feature of these polytopes is that they are simple. For understanding whether there exists a strongly polynomial time algorithm for linear programming, it suffices to study the case of linear programs over rock extensions.

For our purposes, the candidate algorithm we would be interested in is a path following algorithm like the simplex method that traverses the graph of the polytope. Hence, we ask the following question: Can one find a polynomial length path between any pair of vertices of a rock extension in strongly polynomial time? It turns out the answer is yes. However, one needs to be careful with the setup. A rock extension $Q$ is built from a simple polytope $P = \{\mathbf{x} \in \mathbb{R}^{d}: A \mathbf{x} \leq \mathbf{b}\}$, where $A$ is $m \times d$ satisfying strong nondegeneracy assumptions and such that we know a strictly feasible point $o \in P$. The rock extension $Q$ is a $(d+1)$-dimensional simple extended formulation for $P$ with $m+1$ facets and a distinguished vertex $(o,1)$. 

Our argument is essentially a corollary of their proof in  \cite{RockExtensions}. However, the result was important enough context for ours that we include it here together with a proof, and it is not said in their paper. For brevity, we do not completely rewrite their construction here and instead only include the details of it relevant for the consequence we are interested in. We refer the reader to \cite{RockExtensions} for further, more explicit details about the construction.

% \begin{theorem}
% Let $Q$ be a rock extension of a simple polytope with $m+2$ facets in dimension $d+1$. Then there is a weakly polynomial time algorithm to find a path of length at most $2(m-d)$ between any pair of vertices of $Q$. If one further knows $(o,1)$ used in the construction of $Q$, there is a strongly polynomial time algorithm.
% \end{theorem}

\begin{proof}[Proof of Theorem \ref{thm:rockextension}]
To prove this, we need to unpack the proof of Theorem 2.7 of \cite{RockExtensions} for constructing rock extensions. For $\mathbf{u} \in \mathbb{R}^{d}$ and $\varepsilon > 0$, let $B_{\varepsilon}^{d}(\mathbf{u})$ denote the $d$-dimensional open ball of radius $\varepsilon$ centered at $\mathbf{u}$. In order to construct a rock extension, they start with a polytope $P = \{\mathbf{x} \in \mathbb{R}^{d}: A\mathbf{x} \leq \mathbf{b}\}$ and a point $o$ such that $B_{\varepsilon}^{d}(o) \subseteq P$. Then they construct a so-called \emph{rock extension}, which is a simple polytope defined as
\[Q = \{(\mathbf{x},z) \in \mathbb{R}^{d+1}: A \mathbf{x} + \mathbf{y}z \leq \mathbf{b} \text{ and } z \geq 0\}\] where $\mathbf{y}\in \mathbb{R}_{>0}^d$ is a suitably chosen vector which forms part of their construction.

In particular, the projection of $Q$ onto its first $d$ coordinates is exactly $P$. Their construction is then built in such a way that $(o, 1)$ will be a vertex of $Q$ and the unique maximizer for the linear program $\max_{(\mathbf{x},z)\in Q} z$. The way $Q$ is built is inductive by adding one inequality at a time. Initially, up to a reordering of the rows, it is the simplex: 
\[P_{d+1} = \{\mathbf{x} \in \mathbb{R}^{d+1}: A_{[d+1]}\mathbf{x} + \mathbf{y}_{[d+1]} z \leq \mathbf{b}_{[d+1]}, z \geq 0\}.\]
This simplex has one vertex with positive $z$ coordinate, which is exactly $(o,1)$. More generally, 
\[P_{k} = \{\mathbf{x} \in \mathbb{R}^{d+1}: A_{[k]}\mathbf{x} + \mathbf{y}_{[k]} z \leq \mathbf{b}_{[k]}, z \geq 0\}\]
for each $k \geq d+1$. For each $P_{k}$, there is a subset $V_{k}$ of the vertices of $P_{k}$ consisting of all vertices with positive $z$ coordinate. As they complete their construction, they note that there is a sequence of strictly increasing values $0 < \mu_{d+1} < \mu_{d+2} < \dots < \mu_{m} = \varepsilon$ such that $V_{k} \setminus V_{k-1} \subseteq  B_{\mu_{k}}^{d+1}((o,1)) \setminus B_{\mu_{k-1}}^{d+1}((o,1))$. The way they ensure this is by choosing $y_{k}$ such that the hyperplane $H_{k} = \{\mathbf{x} \in \mathbb{R}^{d}| A_{k}\mathbf{x} + y_{k} z= b_{k}\}$ is supporting for the ball $B_{\mu_{k-1}}^{d+1}((o,1))$ and arguing any new vertex created must not be too much further away. 

By virtue of their construction, each vertex in $V_{k}$ is a vertex of the rock extension, and the vertices of the rock extension are $V_{m} \cup V_{m+1}$, where 
\[V_{m+1} = \{(\mathbf{v},0): \mathbf{v} \text{ is a vertex of } Q\}.\]
Every vertex in $V_{m+1}$ is adjacent to a vertex in $V_{m}$. This gives rise to a simple algorithm to find a path of length at most $2(m-d)$ between any pair of vertices of $Q$. To do this, it suffices to find a path of a length at most $m-d$ from any vertex to $(o,1)$ efficiently. For this, simply move to the neighbor that is closest to $(o,1)$. 

Namely, let $\mathbf{v}$ be a vertex of $Q$. Let $\mathbf{v} \in V_{k}$. Then $\mathbf{v}$ has a neighbor in $V_{k-1}$, and any such neighbor is in $B_{\mu_{k-1}}^{d+1}((o,1))$, while any other neighbor is not. Hence, its closest neighbor to $(o,1)$ is in $V_{k-1}$. Since $V_{d+1} = \{(o,1)\}$ this path will reach $(o,1)$ in at most $m-d$ steps. Computing the closest neighbor to $(o,1)$ may be done in strongly polynomial time if $(o,1)$ is known, since by simplicity, each vertex has only $d+1$ neighbors that may be computed using a simplex tableau. If $(o,1)$ is not known, it can be found in weakly polynomial time by the linear program maximizing $z$. Therefore, if $(o,1)$ is known there is a strongly polynomial time algorithm to find a path of length at most $2(m-d)$ between any pair of vertices on $Q$. Otherwise, it can be done in weakly polynomial time.

% If $\mathbf{v} \in V_{m+1}$, the closest neighbor to $(o,1)$ is in $V_{m}$, since any vertex of $V_{m+1}$ is at least $\varepsilon$ away from $(o,1)$, and any vertex of $V_{m}$ is in $B_{\varepsilon}^{d+1}((o,1))$. Let $\mathbf{v}$ be a vertex in $V_{k} \setminus V_{k-1}$. Any neighbor in $V_{k-1}$ must be in $B_{\mu_{k-1}}^{d+1}((o,1))$
\end{proof}

\section{Conclusion}
\label{sec:conclusion}

Knowing that finding shortest paths on a simple polytope is hard does not exclude the possibility that one may find short paths efficiently on general simple polytopes beyond rock extensions. For example, approximation algorithms may be possible. Given the relevant results in the literature, we suspect this is an APX-hard problem and leave proving APX-hardness as an open question. Our argument does not yield any interesting APX-hardness results as one can always efficiently find a path of length one more than the shortest path that we use to model our decision problem.

Finally, the core motivation for understanding these hardness questions is to approach the problem of whether there is a polynomial time version of the simplex method. In particular, one could show the answer is no conditional on $\PP \neq \NP$ by showing that computing a polynomial length path in the graph of a simple polytope is \NP-hard at least with a more standard Phase 1 procedure than that of constructing a rock extension. All hardness results thus far have relied on showing that determining the existence of a short path is hard. However, if the polynomial Hirsch conjecture holds, then a polynomial length path always exists and so this approach could not resolve the question of existence of a polynomial time simplex method. Our final open question is whether one can encode a hard search problem and prove TFNP-hardness for finding a short path on a simple polytope for which we know that short paths exist. This would, in particular, contrast with our observation for rock extensions.

% \section*{AI Disclosure}
% The idea for the proof of  reduction to \textsc{Partition with even sum} came about from conversation with Chat GPT. However, the proof itself was human generated and written, and all text in this paper was written and verified by the human authors.  

\section*{Acknowledgments} We would like to warmly thank Christian N\"{o}bel and Laura Sanit\`a for interesting and helpful discussions on this subject as well as Kirill Kukharenko for help regarding rock extensions. The idea for the proof of the reduction to Partition arose in part from conversations with Chat GPT. However, all the work here was completely written and verified carefully by the (human) authors.

\bibliographystyle{amsplain}
\bibliography{bibliography.bib}

\end{document}